\documentclass[11pt, a4paper]{amsart}
\usepackage[dvips]{epsfig}
\usepackage{amsmath}
\usepackage{amssymb}
\usepackage{amsfonts}
\usepackage{amsthm}
\usepackage{amsbsy}
\usepackage{amsgen}
\usepackage{amscd}
\usepackage{amsopn}
\usepackage{amstext}
\usepackage{amsxtra}
\usepackage{enumerate}
\usepackage{dsfont}
\usepackage{hyperref} 

\newtheorem{theorem}{Theorem}[section]
\newtheorem{proposition}[theorem]{Proposition}
\newtheorem{lemma}[theorem]{Lemma}

\theoremstyle{definition}
\newtheorem{definition}[theorem]{Definition}
\newtheorem{notation}[theorem]{Notation}
\newtheorem{remark}[theorem]{Remark}

\newcommand {\Z} {\mathbb{Z}}

\newcommand {\R} {\mathbb{R}}
\newcommand {\RR} {\mathcal{R}}

\newcommand {\T} {\mathbb{T}}
\newcommand {\E} {\mathbb{E}}

\newcommand {\C} {\mathbb {C}}

\newcommand {\M} {\mathcal {M}}

\newcommand {\Scircle} {\mathcal {S}^{1}}

\newcommand {\eigspcdim} {\mathcal{N}}
\newcommand {\eigval} {E}

\newcommand{\len}{\mathcal {L}}
\newcommand{\var}{\operatorname{Var}}

\newcommand{\tr}{\operatorname{tr}}
\newcommand{\rk}{\operatorname{rk}}

\newcommand{\vol}{\operatorname{Vol}}
\newcommand{\length}{\operatorname{length}}

\newcommand{\sing}{B}
\newcommand{\nonsing}{\T\setminus\sing}

\newcommand{\meas}{\operatorname{meas}}

\def\given{\left.\vphantom{\hbox{\Large (}}\right|}

\usepackage{color}

\begin{document}

\begin{abstract}
Using the spectral multiplicities of the standard torus, we
endow the Laplace eigenspaces with Gaussian probability measures.
This induces a notion of random Gaussian Laplace eigenfunctions
on the torus (``arithmetic random waves'').  We study the  
distribution of
the nodal length of random eigenfunctions for large eigenvalues, 
and our primary result is that the asymptotics for the variance is
{\em non-universal}, and is intimately related to the arithmetic of
lattice points lying on a circle with radius corresponding to the
energy.
\end{abstract}

\title[Nodal length fluctuations]
{Nodal length fluctuations for arithmetic random waves \\
}


 \author{Manjunath Krishnapur}
 \address{Department of Mathematics, Indian Institute of Science, Bangalore 560012}
 \email{manju@math.iisc.ernet.in}
\urladdr{http://math.iisc.ernet.in/\~{}manju}

\author{P\"ar Kurlberg}
\address{Department of Mathematics, Royal Institute of Technology,
 SE-100 44 Stockholm, Sweden}
 \email{kurlberg@math.kth.se}
\urladdr{www.math.kth.se/\~{ }kurlberg}

\author{Igor Wigman}
\address{Cardiff University, Wales, UK}
 \email{wigmani@cardiff.ac.uk}
\urladdr{www.cardiff.ac.uk/maths/subsites/wigman/index.html}

\thanks{M.K.  was supported in part by DST and UGC through DSA-SAP-PhaseIV. \\
P.K.  was  supported in part by grants from
the G\"oran Gustafsson Foundation,
the Knut and Alice Wallenberg foundation, and the Swedish Research
Council. \\ I.W.  was supported in part
by grant KAW 2005.0098 from the Knut and Alice Wallenberg
Foundation, and EPSRC grant under the First Grant Scheme. }


\maketitle


\section{Introduction}



The purpose of this paper is to investigate the variance of the
fluctations of nodal lengths of random Laplace eigenfunctions on the standard 
$2$-torus $\T := \R^{2}/\Z^{2}$.
The nodal set of a function $f$ is simply the zero set of $f$, and
if $f : \T \to \R$ is a  Laplace eigenfunction, i.e., if $f$ is non-constant and  
\begin{equation}
\label{eq:Laplace eqtn}
\Delta f + Ef = 0, \quad E >0,
\end{equation}
then the nodal set of $f$ consists of a union of smooth curves outside a finite
set of singular points (see ~\cite{Cheng}) and hence
$\length(f^{-1}(0))$, the {\em nodal length of $f$}, is well defined.

A fundamental conjecture by Yau \cite{Y1,Y2} asserts that for any
smooth compact Riemannian manifold $M$, there exist constants
$c_{2}(M)\ge c_{1}(M)>0$ such that
\begin{equation}
\label{eq:Yau conj}
c_{1}(M) \cdot \sqrt{E } \leq \vol(f^{-1}(0)) \leq c_{2}(M) \cdot \sqrt{E }
\end{equation}
for  any Laplace eigenfunction $f$ on
$M$ with eigenvalue $E$.
By the work of Donnelly \& Fefferman \cite{Donnelly-Fefferman} and  
Br\"uning \& Gromes \cite{Bruning,Bruning-Gromes}, Yau's conjecture  
is known to be true for manifolds with real analytic metrics, and in particular 
for $M=\T$.


For $\T= \R^{2}/\Z^{2}$ the sequence of eigenvalues, or energy levels,
are related to integers  expressible as a sum of two integer squares;
if we define
$
S := \{ n :\: n=a^2+b^2,\, a,b\in\Z \}, 
$
the eigenvalues are all of the form
\begin{equation}
\label{eq:eigval=4pi^2 n}
\eigval_{n} := 4\pi^2 n, \quad n \in S.
\end{equation}
For $n \in S$, let
\begin{equation*}
\Lambda_{n} :=
\left\{ \lambda\in\Z^{2}: \|\lambda \|^2 = n  \right\}
\end{equation*}
denote the corresponding frequency set.
Using the standard notation
$e(z):= \exp(2\pi i z)$, the $\C$-eigenspace $\mathcal{E}_{n}$
corresponding to $E_{n}$ is
spanned by the $L^{2}$-orthonormal set of functions
$\left\{ e\left(\langle \lambda,\, x \rangle\right)
\right\}_{\lambda\in\Lambda_{n}}$.
The dimension of $\mathcal{E}_{n}$, denoted by
\begin{equation*}
\eigspcdim_{n} := \dim \mathcal{E}_{n} = r_{2}(n) = |\Lambda_{n}|
\end{equation*}
is equal to the number $r_{2}(n)$ of different ways $n$ may be
expressed as a sum of two squares.

\subsection{Arithmetic random waves}
\label{sec:arithm-rand-waves}

The set $\Lambda_n$ can be identified with the set of
lattice points lying on a circle with radius $\sqrt{n}$, and its properties
are intimately
related to representations of integers by the quadratic form
$x^{2}+y^{2}$.  The frequency set is thus of arithmetic nature. A
particular consequence is that the sequence of spectral multiplicities
$\{\eigspcdim_{n}\}_{n \geq 1}$ is {\em unbounded}.
It is thus natural to consider properties of ``generic'', or
``random'', eigenfunctions $f_{n} \in \mathcal{E}_{n}$, and our
primary interest is the high energy asymptotics of the distribution of
their nodal length $\len(f_{n}) $ as $n$ tends to
infinity in such a way that $\eigspcdim_{n} \to \infty$.
More precisely,
let $f_{n}:\T\rightarrow\R$ be the random Gaussian field
of (real valued) $\mathcal{E}_{n}$-functions with eigenvalue
$\eigval_{n}$, i.e.,
\begin{equation}
\label{eq:fn def}
f_{n}(x) = \frac{1}{\sqrt{2\eigspcdim_{n}}} \sum\limits_{\lambda\in\Lambda_{n}} a_{\lambda}
e\left(\langle \lambda,\, x \rangle\right),
\end{equation}
where $a_{\lambda} = b_{\lambda}+i c_{\lambda}$ are independent standard complex
Gaussian random variables, 
save for the relations $a_{-\lambda} = {\overline a_{\lambda}}$. This just means that
$b_{\lambda},c_{\lambda} \sim \mathcal{N}(0,1)$  are standard real Gaussians satisfying the relation
$b_{-\lambda}=b_{\lambda}$, $c_{-\lambda} = - c_{\lambda}$ and
otherwise independent. Our object of study is the random variable
\begin{equation*}
\len_{n} := \len (f_{n}) = \length(f_{n}^{-1}(0)),
\end{equation*}
henceforth called the {\em nodal length} of $f_{n}$. 

\subsection{Prior work on this model} In this setting, Rudnick and Wigman ~\cite{RW} computed the expected
nodal length of $f_{n}$ to be $ \E[\len_{n}] =
\frac{1}{2\sqrt{2}} \cdot \sqrt{E_{n}}$, in agreement with Yau's conjecture 
\eqref{eq:Yau conj}. They also gave the bound
\begin{equation}
\label{eq:var << E/sqrt(N)}
\var(\len_{n}) = O\left(\frac{\eigval_{n}}{\sqrt{\eigspcdim_{n}}}
\right)
\end{equation}
for the variance, and conjectured that the stronger bound
\begin{equation}
\label{eq:var=O(E/N) conj}
\var(\len_{n}) = O\left(\frac{\eigval_{n}}{\eigspcdim_{n}} \right)
\end{equation}
holds. A nice consequence of \eqref{eq:var << E/sqrt(N)} is that $\len(f_{n})$ concentrates around its mean. More precisely, there is a sequence $\delta_{n}\to 0$ such that
$$\mathbf P\left((1-\delta_{n})\frac{\sqrt{E_{n}}}{2\sqrt{2}}\le \len(f_{n})  \le (1+\delta_{n})\frac{\sqrt{E_{n}}}{2\sqrt{2}} \right) \to 1
\quad \text{as $\eigspcdim_{n} \to \infty$.}
$$  


In this paper we shall determine the leading order asymptotic of
$\var(\len_{n})$ as $\eigspcdim_{n} \to \infty$. As consequence we
improve on the conjectured bound
\eqref{eq:var=O(E/N) conj} and obtain
the {\em sharp} bounds
\begin{equation*}
\frac{\eigval_{n}}{\eigspcdim_{n}^{2}} \ll \var(\len_{n}) \ll \frac{\eigval_{n}}{\eigspcdim_{n}^{2}}.
\end{equation*}
It turns out that the asymptotic behaviour of the variance is {\em
  non-universal} in the sense
that it depends on the angular distribution of the points in the
frequency set $\Lambda_{n}$. In the proof, a  leading order sum involving many terms of size
$\eigval_{n}/\eigspcdim_{n}$ surprisingly {\em cancels perfectly}, and the
variance is   therefore much smaller than expected! We may say that
$\T$ exhibits ``arithmetic Berry cancellation''  (cf.
Section~\ref{sec:berrys-canc-phen}). 

\subsection{Our results} In order to describe our results we shall need
some further notation.
The set $\Lambda_{n}$ induces a discrete probability measure
$\mu_n$ on the circle $\Scircle = \{z\in\C:\: |z|=1\}$ by defining
\begin{equation}
\label{eq:mun def}
\mu_{n} := \frac{1}{\eigspcdim_{n}}
\sum\limits_{\lambda\in\Lambda_{n}} \delta_{\frac{\lambda}{\sqrt{n}}},
\end{equation}
where $\delta_{x}$ is the Dirac delta measure supported at $x$.
%
%
The Fourier transform of $\mu_n$ is, for any $k \in \Z$, as usual
given by
$
\widehat{\mu}_{n}(k) := \int_{\Scircle} z^{-k}d\mu_{n}(z).
$
For $n \in S$, we define
\begin{equation}
\label{eq:c(mu) def}
c_{n} := \frac{1+\widehat{\mu}_{n}(4)^2}{512};
\end{equation}
it is then easy to see that $c_{n}$ is real and that $c_{n} \in
[1/512,1/256]$.  (Since $\Lambda_{n}$ is invariant under the
transformations $z \to \overline{z}$ and $z \to i\cdot z$, the same
holds for $\mu_n$, hence $\widehat{\mu}_{n}(4) \in \R$.  Further,
since $\mu_n$ is a probability measure, $|\widehat{\mu}_{n}(4)| \leq
1$, and consequently $\widehat{\mu}_{n}(4)^{2} \in [0,1]$.)
%

%

We can now formulate our principal result.
\begin{theorem}
\label{thm:var(L) asymp united}
If $(n_{i})_{i\geq 1}$ is any sequence of elements in $S$ such that
$\eigspcdim_{n_i} \to \infty$, then
\begin{equation}
\label{eq:var(L) asymp united}
\var\left(\len_{n_{i}}\right) = c_{n_{i}} \cdot
\frac{\eigval_{n_{i}}}{\eigspcdim_{n_{i}}^2} (1+ o(1) ).
\end{equation}
Further, given any $c \in [1/512,1/256]$, there exists a sequence
$(n_{i})_{i\geq 1}$ of elements in $S$ such that as $i \to \infty$,
we have $\eigspcdim_{n_i} \to \infty$ together with $c_{n_{i}}\to c$
so that
$$
\var\left(\len_{n_{i}}\right) = c \cdot
\frac{\eigval_{n_{i}}}{\eigspcdim_{n_{i}}^2} (1+ o(1) ).
$$
\end{theorem}

\subsection{Attainable measures} The second part of the theorem, in light of the
first one, amounts to the following: given any $\alpha \in [0,1]$,
there exists a sequence $(n_{i})_{i\geq 1}$ such that
$\widehat{\mu}_{n_{i}}(4)^{2} \to \alpha$.
We briefly describe the measures $\mu_n$ giving rise to the extremal
points as well intermediate values attainable by $c_{n}$
(cf. Section~\ref{sec:prob-meas-scircle} for full details, in
particular the precise notions of generic and thin used below.)

It is well known that the lattice points $\Lambda_{n}$ are
equidistributed
on $\Scircle$ along generic subsequences of energy
levels, see e.g. \cite{FKW}, Proposition 6.  Thus, for
$(n_{i})_{i\geq 1}$ a {\em generic} sequence of elements in $S$ the
variance is {\em minimal} in the limit since
$\widehat{\mu}_{n_{i}}(4) \to 0$, and thus $c_{n_{i}}\to 1/512$.
It is also worthwhile mentioning that the nodal length variance
of $f_{n}$ for such generic sequences differs by an order of magnitude
from the corresponding quantity for superposition of random
planar waves with same wavelength
and directions chosen uniformly on the unit circle
(especially noteworthy in light of the fact that both have the same {\em scaling}
properties.  
 See the last paragraph of Section \ref{sec:RWM vs torus vs S2} for a more detailed
explanation).

As for the maximum, Cilleruelo ~\cite{Cil} has shown that there
are thin sequences  $(n_{i})_{i\geq 1}$, with
$\eigspcdim_{n_{i}}\rightarrow\infty$, such that $\mu_{n_{i}}$
converges weakly to an atomic probability measure supported at the $4$
symmetric points $\pm 1$, $\pm i$; hence $\widehat{\mu}_{n_{i}}(4)
\to 1$ and $c_{n_{i}} \to 1/256$.
For the intermediate values
we construct thin sequences $(n_{i})_{i\geq1}$ of elements in $S$
such that $\mu_{n_i}$ converges weakly to the uniform probability
measure supported on a union of four arcs.

More precisely, for $a\in [0,\frac{\pi}{4}]$ define a probability measure  $\nu_{a}$ on $\Scircle$ by
\begin{equation}
\label{eq:mua def}
\nu_{a} := \left(\frac{1}{4}\sum\limits_{k=0}^{3}\delta_{i^{k}} \right) \star
\tilde{\nu}_{a}
\end{equation}
where $\star$ stands for convolutions of measures, and
$\tilde{\nu}_{a}$ is the uniform measure on $[-a,a]$ (identifying
$\Scircle\cong \R/2\pi\Z$).  More explicitly,
\begin{equation*} 
\tilde{\nu}_{a}(f) = \frac{1}{2a}\int_{-a}^{a} f\left(e^{i\theta}\right) d\theta, \qquad
\nu_{a}(f) = \frac{1}{8a}\sum\limits_{k=0}^{3}\int\limits_{-a+k\frac{\pi}{2}}^{a+k\frac{\pi}{2}} f\left(e^{i\theta}\right)d\theta.
\end{equation*}
For $a=0$, we shall use the notational convention that $\nu_{0} =
\frac{1}{4}\sum_{k=0}^{3}\delta_{i^{k}}$.
\begin{proposition}
\label{prop:intermd meas exist}
For every $a\in [0,\frac{\pi}{4}]$ there exists a sequence $\eigval_{n_{i}}$ of energy levels, such that
$\mu_{n_{i}} \Rightarrow \nu_{a}$ with $\nu_{a}$ as in \eqref{eq:mua def}.
In particular,
for every $b\in [0,1]$, there exists a  sequence $\eigval_{n_{i}}$ of energy levels such that
$c_{n_{i}} \rightarrow \frac{1+b}{512}.$
\end{proposition}
Note that the second statement follows from the first because the values of $\widehat{\nu}_{a}(4)$
ranges
over the whole of $[0,1]$ as $a$ ranges over $[0,\frac{\pi}{4}]$
(in fact, an easy computation shows that
$\widehat{\nu}_{a}(4)=\frac{\sin(4a)}{4a}$.)
Further, the extremal values $b=0$ and $b=1$ are attained by $\nu=\nu_{\frac{\pi}{4}}$, the uniform measure on $\Scircle$, and
$\nu=\nu_{0}$, the atomic symmetrized measure. The proof of Proposition~\ref{prop:intermd meas exist} will be given in Section~\ref{sec:prob-meas-scircle}.



\subsection{Independence of eigenbasis choice and the covariance
  function} 
The random field \eqref{eq:fn def} is centered, Gaussian and stationary in the sense that for any $x_{1},\ldots, x_{k}\in\T$
and $y\in\T$, the random vector $$\left(f_{n}(x_{1}+y),\ldots ,f_{n}(x_{k}+y)
\right)\in\R^{k}$$ is a mean zero multivariate Gaussian, whose
distribution does not depend on $y$.
The covariance
function\footnote{The covariance function is widely
referred to as the $2$-point function in the physics literature.} $
r(x)=r_{n}(x) := \E[f_{n}(y)f_{n}(x+y)]
$ 
thus depends only on $x$, and we may express it explicitly as
\begin{equation}
\label{eq:r def}
r_{n}(x) = \frac{1}{\eigspcdim_{n}}\sum\limits_{\lambda\in\Lambda}
e\left(\langle \lambda,\, x \rangle  \right) = \frac{1}{\eigspcdim_{n}}\sum\limits_{\lambda\in\Lambda}\cos\left( 2\pi \langle \lambda,\, x \rangle  \right).
\end{equation}

Though the normalizing factor
in the definition \eqref{eq:fn def} of $f_{n}$
has no bearing on the nodal length, it is convenient to work with, and
we have chosen to have  
$r_{n}(0)=1$, or, equivalently, for every
$x\in\T$,  $\E[f_{n}(x)^2] = 1$. 

The  covariance function {\em determines} the distribution of a centered Gaussian random
field, and, in principle,
one may express any aspect of the geometry of $f_{n}$
in terms of $r_{n}$ only (cf. Kolmogorov's Theorem,
~\cite{CL} Chapter 3.3). 
This important fact also shows that we would get the same  random field in \eqref{eq:fn def} had we chosen a different orthonormal basis of $\mathcal{E}_{n}$ in the Gaussian linear combination.  

\subsection{Background and results in related models} 
\label{sec:background}

The question of distribution of various local quantities such as the
nodal length, or the total curvature of nodal lines in different settings,
has been  extensively studied.
It is widely believed ~\cite{Berry 1977} that for {\em generic chaotic} billiards, one can model the nodal lines for
eigenfunctions of eigenvalue of order $\approx E$ with
nodal lines of planar monochromatic random waves of wavenumber $\sqrt{E}$
(this is called Berry's Random Wave Model or RWM, see \eqref{eq:RWMcovariance} for the definition). 
 Hence the importance of the
(random) nodal length and other properties of the RWM.
Berry \cite{Berry 2002} found
that the expected nodal length (per unit area) for the RWM is of size
approximately $\sqrt{E}$, and argued that the variance should be of order $\log E$.

The $2$-dimensional unit sphere $\mathcal{S}^{2}$ is another manifold
with degenerate Laplace spectrum. Here
the Laplace eigenvalues are all the numbers
$E=m(m+1)$ with $m \geq 0$ an integer, and the corresponding
eigenspace is the space of degree $m$ spherical harmonics; its
dimension is $2m+1$.
One may define the random field of degree $m$ spherical harmonics
similarly to the torus
\eqref{eq:fn def} with the plane waves (exponentials) replaced by any 
$L^{2}$-orthonormal basis $\{\eta_{m;1},\ldots \eta_{m;2m+1}\}$ of real valued spherical harmonics
of degree $m$: 
\begin{equation*}
f_{m}^{\mathcal{S}^{2}}(x) = \frac{1}{\sqrt{2m+1}}\sum\limits_{k=1}^{2m+1}a_{k}\eta_{m;k}(x)
\end{equation*}
with $a_{k}$ i.i.d standard Gaussian.  
Setting $\len \left(f_{m}^{\mathcal{S}^{2}}\right)$
to be the nodal length of $f_{m}^{\mathcal{S}^{2}}$, Berard ~\cite{Berard}
computed the expected nodal length
\begin{equation*}
\E \left[\len \left(f_{m}^{\mathcal{S}^{2}}\right) \right] = \sqrt{2}\pi \cdot \sqrt{E}.
\end{equation*}
Wigman ~\cite{W sphere} found that the nodal length variance is asymptotic to
\begin{equation*}
\var \left(\len \left(f_{m}^{\mathcal{S}^{2}}\right)\right) \sim c\log{m}, 
\end{equation*}
which is consistent with Berry's prediction for the RWM.

\subsubsection{Comparing the random wave model to the torus and the sphere}
\label{sec:RWM vs torus vs S2}

The logarithmic variance is much smaller than one would expect: taking
into account that the wavelength for either the sphere or the RWM scales as
$\frac{1}{\sqrt{E}}$,
one may rescale them to unit wavelength to argue that the nodal
length variance should be proportional to $\sqrt{E}$.
However, a computation reveals that
the coefficient in front of the expected leading term $\sqrt{E}$
surprisingly {\em vanishes} due to ``Berry's
Cancellation Phenomenon" --- the leading term for nodal length
variance is in fact logarithmic.  A similar
cancellation phenomenon is responsible for the variance
\eqref{eq:var(L) asymp united} in our situation being
of
order of magnitude $\frac{\eigval_{n}}{\eigspcdim_{n}^{2}}$, rather
than of order
$\frac{\eigval_{n}}{\eigspcdim_{n}}$ as was originally conjectured in 
~\cite{RW}.

As already remarked, in general, defining a centered
(or mean zero) Gaussian random field $f$ on an arbitrary domain 
$T$ is equivalent to specifying
its covariance function $r_{f}(x,y):=\E[f(x)f(y)]$ on $T\times T$. 
For the planar random waves (RWM) the covariance
function is
\begin{equation}\label{eq:RWMcovariance}
r_{\text{RWM}}(x,y)=J_{0}(\sqrt{E}\|x-y\|) 
\end{equation}
with $J_{0}$ the standard Bessel function,
and $$r_{\mathcal{S}^{2}}(x,y)=P_{m}(\cos {d(x,y)})$$
for the degree $m$ spherical harmonics, where $P_{m}$ are the usual Legendre polynomials, and $d$ is the
(spherical) distance. The latter scales as
\begin{equation*}
P_{m}(\cos (\psi/m))\approx J_{0}(d),
\end{equation*}
uniformly for $\psi\in \left[0,m\cdot \frac{\pi}{2}\right]$. As
the corresponding eigenvalues are $m(m+1)$,
this is consistent to the RWM scaling.
The covariance function  
$r_{n}(x)$ for our ensemble $f_{n}$ of random toral eigenfunctions
given by \eqref{eq:fn def} is of arithmetic flavour: it is given by the summation \eqref{eq:r def}
over lattice points $\Lambda_{n}$ lying on a circle.

The equidistribution of $\Lambda_{n}$ along generic sequences of energy levels on the torus
mentioned earlier implies that for any {\em fixed} $y\in\R^{2}$ one may approximate
\begin{equation*}
r_{n}(y/(2\pi\sqrt{n_{i}})) \approx \int\limits_{\Scircle}\cos(\langle y,z \rangle)dz = J_{0}(\|y\|)
\end{equation*}
for a generic sequence $\{ n_{i}\}\subseteq S$. Although it is the same scaling limit as before,
the latter holds for $y$ of fixed size only, and 
by no means uniformly for $y\in [0,n]^{2}$. In particular, as opposed to the other cases, no ``intermediate
range asymptotic" for $r_{n}(x)$ is known, i.e. for $x\cdot \sqrt{n} \rightarrow\infty$.
It is remarkable that even in this case, in spite of the fact that
the covariance function for $f_{n}$ has the same scaling limit as
the RWM and random spherical harmonics random fields,
the nodal length variance \eqref{eq:var(L) asymp united} of random arithmetic waves
is of different order of magnitude compared to
the other cases.

\subsubsection{Berry's cancellation phenomenon and the $2$-point correlation function}

\label{sec:berrys-canc-phen} 

In order to evaluate moments of the nodal length of a random field $f$ we exploit a suitable
Kac-Rice type formula (see Section~\ref{sec:kacrice}). For the variance it means that we need to understand the
fluctuations of the so called $2$-point correlation function (defined on the
domain of $f$) around its scaled asymptotic
value at infinity.
For both $\R^{2}$ (RWM) and $\mathcal{S}^{2}$ the main contribution for the variance comes
from the ``intermediate range'' (i.e. a few wavelengths away from the origin), where the asymptotic behaviour
of the covariance function and its first two derivatives translates 
to asymptotics for the $2$-point correlation function. No analogous asymptotics is known for the torus.
The cancellation phenomenon amounts to the fact that the leading term in the
intermediate range asymptotics for the $2$-point correlation function is purely oscilliatory
and its contribution to the integral is negligible. Then, the main contribution comes
from the second term in the asymptotics. 

As a substitute for pointwise
asymptotics for the $2$-point correlation function, we use 
an arithmetic formula (see \eqref{eq:K2 asymp r,X,Y}) valid outside a suitably defined
``singular set" (arithmetic in nature; its
analogue for $\R^{2}$ and sphere is a neighbourhood of the origin, with radius
of order wavelength). Although the arithmetic formula does not give the pointwise
behaviour of the $2$-point correlation function, its arithmetic structure is exploited
for averaging over the torus and is essential for  evaluating the variance. 
The ``arithmetic
Berry's cancellation'' amounts to the Fourier
expansion of the highest magnitude term of the $2$-point correlation function
vanishing at the origin, an artifact of the seemingly unrelated
trigonometric identity $$ 4\cos(\theta/2)^{4} = 1+2\cos\theta+\cos(\theta)^{2}  $$
(see Section \ref{sec:discussion Berry's cancellation} for more details).

\subsection{Some other related results}

For a generic compact manifold $\M$ with no spectral degeneracies,
one can also consider random Gaussian linear
combinations of eigenfunctions
with different eigenvalues (sometimes referred to as ``random wave on $\M$'').
Berard ~\cite{Berard} and Zelditch ~\cite{Z} found that, given a spectral
parameter $E$, the expected nodal length for random Gaussian superpositions of eigenfunctions
with eigenvalues lying either about $E$ or below\footnote{Called the short or long energy window random combinations respectively.}
$E$ is of order $\sqrt{E}$, consistent
with Berry's RWM. The subtle question of the nodal length variance in this generic setup is
to be addressed in  ~\cite{PW}.

Some other generic
results concerning random waves with spectral parameter $E$:
Toth and Wigman ~\cite{TW1}
found that the expected number of open nodal lines, i.e., the connected
component of the zero set that intersect the boundary, of
the random wave with spectral parameter $E$ on a generic surface with boundary
is again of order $\sqrt{E}$.
Moreover, Nicolaescu ~\cite{NL} evaluated the expected number of critical points to be of order
$E$; the latter is also an upper bound for the number of nodal domains.

For other related or relevant result we refer the interested reader to
the recent survey ~\cite{W Dartmouth}. Also, for recent very
interesting results and conjectures on non-local quantities, such as
{\em nodal domains} (i.e. the connected components of the complement
$\M\setminus f^{-1}(0)$ of the nodal line) see  ~\cite{BS,NS,BR1,BR2,BR3}.

\subsection{Outline of the paper}

The paper is organized as follows.
The proof of Theorem \ref{thm:var(L) asymp united} assuming certain preparatory results
is given in Section \ref{sec:proof outline}; this proof
relies on an arithmetic formula \eqref{eq:var as hat(mu)(4)} in Proposition \ref{prop:var as hat(mu)(4)},
whose proof is commenced in Section \ref{sec:proof-arithmetic-prop}.
The proof of the formula \eqref{eq:var as hat(mu)(4)} is based on
studying the behaviour of the so-called $2$-point correlation function introduced in
Section \ref{sec:2-point correlation}; its subtle asymptotic analysis is given
in sections \ref{sec:proof-arithmetic-prop}, \ref{sec:asymptotics2ptcorelation}
and \ref{apx:proof asymp r,D,H}, with the latter containing a certain technical
computation essential for understanding the asymptotic properties.

Section \ref{sec:proof of S6=o(N^4)} is dedicated to the proof of Theorem
\ref{thm:arith corr S6}, an arithmetic bound, due to Bourgain, needed for the admissability of
the error term in \eqref{eq:var as hat(mu)(4)}.
In Section \ref{sec:prob-meas-scircle} sequences of energy
levels with corresponding discrete probability measures \eqref{eq:mun def} converging to the measures $\nu_{a}$
as in \eqref{eq:mua def} are constructed to prove the attainability of the latter.

\subsection{Acknowledgements}

The authors would like to thank Ze\'{e}v Rudnick for suggesting the
problem as well as having many inspiring and fruitful discussions. In
addition, we would like to thank Mikhail Sodin, Peter Sarnak and Jean
Bourgain for stimulating and fruitful discussions on various aspects of the problem
solved in this manuscript.  
Part of the research
presented was conducted during or as a result of the visit of I.W. to
the Institute for Advanced Study, Princeton, and he would like to
acknowledge the friendly environment, and the generous financial
support.  We are especially grateful to Jean Bourgain for his kind
permission to include Theorem~\ref{thm:arith corr S6} and its
proof.

\section{Proof of Theorem \ref{thm:var(L) asymp united}}
\label{sec:proof outline}
\subsection{Kac-Rice formulas}\label{sec:kacrice} Moments of the nodal length for smooth random fields can be computed using the Kac-Rice formulas ~\cite{CL}. To state them, we need some notation. For $f=f_{n}$, we define its first and second correlations as follows:
\begin{align}
K_{1}&=\E\left[\|\nabla f(y)\| \given  f(y)=0\right],\nonumber\\
\tilde{K}_{2}(x)& =\E\left[\|\nabla f(y)\|\cdot \|\nabla f(y+x)\| \given
f(y)=f(y+x)=0 \right],\label{eq:kacricemoments} \\
K_{2}(x)&= \frac{2}{E_{n}}\tilde{K}_{2}(x).  \nonumber
\end{align}
Observe that $K_{1}$ and $K_{2}$ are independent of $y$ because
$f_{n}$ is stationary (for general smooth Gaussian fields, they become
$K_{1}(y)$ and $\tilde{K}_{2}(x,y)$). They are called the first and
second correlations of the nodal set $f_{n}^{-1}(0)$. $K_{2}$ is just
a scaled version of the second correlation. As we are dealing with
Gaussians, it is possible to write analytical expressions for these as
Gaussian integrals in terms of $r_{n}$ and its derivatives (see
\eqref{eq:K2 as E[|W1| |W2|]}, \eqref{eq:Omega def} and
\eqref{eq:tilde{Omega} def}).

Then, the Kac-Rice formulas say that
\begin{equation*}
\E[\len_{n}] = \int_{\T}K_{1}dy= K_{1},  \qquad
\E[\len_{n}^{2}] = \int\limits_{\T}  \tilde{K}_{2}(x) dx.
\end{equation*}
The first of these formulas gives $ \E[\len_{n}] =
\frac{1}{2\sqrt{2}} \cdot \sqrt{E_{n}}$ as was quoted earlier. Using this and the second gives
\begin{align}
\var(\len_{n}) &= \frac{\eigval_{n}}{2} \int\limits_{\T} \left( K_{2}(x) - \frac{1}{4} \right) dx. \label{eq:kacricevariance}
\end{align}
Full justification of their validity in our context may be found in~\cite{RW}. For the Kac-Rice formulas in general, consult~\cite{CL}. 


It is instructive to intuitively understand the function $\tilde{K}_{2}(x)$ in the following way. 
Let $x\in\T$, and take a small positive number $\epsilon >0$.
We define the random variables $$\len_{n}^{x,\epsilon} = \length \left( f_{n}^{-1}(0)
\cap B(x,\epsilon)  \right),$$ where $B(x,\epsilon)$ is the disk of radius $\epsilon$ centered at $x$; $\len_{n}^{x,\epsilon}$ measures the nodal length of $f_{n}$ inside 
the corresponding disk.
Then we have
\begin{equation*}
\tilde{K}_{2}(x) = \frac{2}{\eigval_{n}}\lim\limits_{\epsilon \rightarrow 0}\frac{1}{\pi^2\epsilon^4}
\E \left[ \len_{n}^{x,\epsilon} \len_{n}^{0,\epsilon}   \right].
\end{equation*}

\subsection{Computing the variance - the Gaussian integral} To understand $\var(\len_{n})$
we need to understand the integral in \eqref{eq:kacricevariance}.
The function $K_{2}$ may be implicitly expressed in terms of the  covariance function
$r_{n}$ of $f_{n}$ as a Gaussian expectation of a $4$-variate
centered Gaussian $(\nabla f_{n}(0), \nabla f_{n}(x))$ conditioned on 
$f_{n}(0)=f_{n}(x)=0,$ with $4\times 4$ covariance matrix $\Omega_{n}(x)$ depending on $r$ and its derivatives
(see \eqref{eq:K2 as E[|W1| |W2|]}, \eqref{eq:Omega def} and \eqref{eq:tilde{Omega} def}).
In our case the covariance function is the arithmetic function \eqref{eq:r def}.


In order to study the asymptotic behaviour of the integral above, we
use some ideas from ~\cite{ORW,RW}
and divide the torus into a
singular set $\sing$ and the nonsingular complement $\nonsing$; only the
latter is convenient to work with, so it is essential
make the former as small as possible.
We improve the analysis of the earlier paper ~\cite{RW} on both
$\sing$ and $\nonsing$.

A better upper bound for the measure of $\sing$ is proved using the $6$th moment of $r$ rather than the
$4$th one.
As an artifact of the definition of $\sing$,
one has a lower bound for the values $|r(x)|$ on
$\sing$; using a Chebyshev-like inequality 
on the $6$-th moment of $r(x)$, we will bound the measure of $\sing$,
so that its contribution to the variance is negligible.

On $\nonsing$, where the main contribution comes from, we establish a precise asymptotic expression for the
$2$-point correlation function compared to a partial upper bound as in ~\cite{RW}
(see Proposition \ref{prop:$2$-pnt using covar nonsing}).
Here the (scaled) covariance matrix $\Omega_{n}(x)$, defined by
\eqref{eq:Omega def},
is a perturbation
of the identity matrix. In order to understand its contribution to the integral, we expand $K_{2}(x)$ as a function
of $\Omega_{n}$ into a $4$-variate Taylor polynomial around $\Omega_{n}=I_{4}$, the identity matrix.
In principle, this could be performed using brute force; we choose to work with Berry's elegant method
~\cite{Berry 2002}.

The computation above culminates in the ``arithmetic formula'' given
in the following proposition. It is of arithmetic essence and at the
heart of the variance being non-universal, but the derivation itself
involves no arithmetic.
%
Before stating the proposition, we define
\begin{equation}
\label{eq:R def}
\mathcal{R}_{k}(n) := \int\limits_{\T}|r_{n}(x)|^{k}dx.
\end{equation}
%
\begin{proposition}
\label{prop:var as hat(mu)(4)}
The nodal length variance is given by the asymptotic formula
\begin{equation}
\label{eq:var as hat(mu)(4)}
\var(\len_{n})   =   c_{n} \cdot \frac{\eigval_{n}}{\eigspcdim_{n}^2}
+ O\left( \eigval_{n} \cdot \mathcal{R}_{5}(n) \right),
\end{equation}
where we used the notation \eqref{eq:c(mu) def}
for $c_{n}$. 
\end{proposition}
The reason we refer to \eqref{eq:var as hat(mu)(4)} as ``arithmetic" is that
both the main term and the error term in \eqref{eq:var as hat(mu)(4)}
are of arithmetic nature: $c_{n}$ is related
to the distribution of lattice points $\Lambda_{n}$ on the circle (see Section
\ref{apx:proof asymp r,D,H}), and $\mathcal{R}_{5}$ is controlled in terms of arithmetics of spectrum
correlation.

The proof of Proposition~\ref{prop:var as hat(mu)(4)} will commence in Section~\ref{sec:proof-arithmetic-prop}.
It is lengthy and quite technical, so it may be omitted on a first reading of the paper.

In case of random spherical harmonics or the random wave model, one arrives at analogous propositions for the variance. The proof essentially ends there as the error term may be checked to be smaller that the main term.

\subsection{Computing the variance - the arithmetic part}
In our setting however, the main obstacle  is in proving the admissability of
the error term.
%
We will control various error terms in terms of the moments
$\mathcal{R}_{k}(n) := \int_{\T}|r_{n}(x)|^{k}dx$ (cf. \eqref{eq:R def}.)
The even moments are naturally related to the spectral correlations;
for example, it is straightforward to check that
\begin{equation}
\label{eq:int r^6=1/N^6 S6(n)}
\RR_{6} (n) = \frac{1}{\eigspcdim_{n}^6}| S_{6}(n)|,
\end{equation}
where $S_{6}$ is the $6$-correlation set of frequencies
\begin{equation}
\label{eq:S6 def}
S_{6}(n) = \left\{(\lambda_{1},\ldots \lambda_{6})\in \Lambda_{n} : \: \sum\limits_{i=1}^{6}\lambda_{i} = 0 \right\}.
\end{equation}

Since for any choice of $\lambda_{1},\ldots,\lambda_{4}\in\Lambda_{n}$
there are at most $4$ possible
choices for $\lambda_{5},\, \lambda_{6} \in \Lambda_{n}$, it follows that
$|S_{6}(n)| = O\left( \eigspcdim_{n}^4  \right),$ or, equivalently
\begin{equation}
\label{eq:R6=O(1/N^2)}
\RR_{6} (n) = O\left(  \frac{1}{\eigspcdim_{n}^2}  \right).
\end{equation}
The latter bound is not quite sufficient for our
purposes, but the following result, due to J. Bourgain
is sufficiently strong for our purposes.
\begin{theorem}
\label{thm:arith corr S6}
As $\eigspcdim_{n}\rightarrow\infty$, we have the following estimate:
\begin{equation}
\label{eq:S6-estimate}
|S_{6}(n)| = o\left(  \eigspcdim_{n}^4  \right).
\end{equation}
Consequently, $\RR_{6} (n) = o\left(  \frac{1}{\eigspcdim_{n}^2}  \right)$.
\end{theorem}
Theorem \ref{thm:arith corr S6} will be proven in Section
\ref{sec:proof of S6=o(N^4)}.

We give a brief preview of the proof
of Theorem~\ref{thm:arith corr S6}. To refute the possibility that $|S_{6}(n)| \gg \eigspcdim_{n}^{4}$,
we invoke some
techniques from additive combinatorics (see Section \ref{sec:proof of S6=o(N^4)}).
Utilizing a notion of ``additive energy" defined in Section \ref{sec:proof
  of S6=o(N^4)},
a certain set $A$ related to the sum set of $\Lambda_{n}$ is shown to contain
a large subset $A_{1}$ with ``bounded doubling".  Using a suitable version
of Freiman's Theorem, this implies that $A_{1}$ is essentially a {\em generalized arithmetic progression}
(GAP - see Theorem \ref{thm:simple-freiman}), i.e., is contained inside a slightly larger GAP.
This then leads to a contradiction via an application of Chang's result
~\cite{chang03-factorization-erdos-szemeredi} on the number of representations of a complex
number as a product of elements inside a GAP.

We note that the bound $|S_{6}(n)| =o(\eigspcdim_{n}^{4} )$ is
quite far from the truth; Bombieri and Bourgain ~\cite{Bombieri-Bourgain} have recently
obtained an exponent savings.
We further remark that a trivial lower bound is that $|S_{6}(n)| \gg
\eigspcdim_{n}^3$; taking any $\lambda_{1},\lambda_{2}, \lambda_{3}
\in \Lambda_n$ and letting
$\lambda_{4}=-\lambda_{1},\lambda_{5}=-\lambda_{2}$, and
$\lambda_{6}=-\lambda_{3}$ yields
$|\Lambda_n|^{3}=\eigspcdim_{n}^{3}$ solutions of
``diagonal type''.
We believe that essentially all solutions arise like this, and
conjecture
that for every $\epsilon > 0$,
$$
|S_{6}(n)| = O_{\epsilon} \left(  \eigspcdim_{n}^{3+\epsilon}
\right).
$$
(Possibly the even stronger bound
$
|S_{6}(n)| = O\left(  \eigspcdim_{n}^{3}
\right)
$
holds.)


\subsection{Proof of Theorem \ref{thm:var(L) asymp united} assuming the preparatory results}

\label{sec:proof thm var united}
Given Proposition~\ref{prop:intermd meas exist},
Proposition~\ref{prop:var as hat(mu)(4)}, and Theorem~\ref{thm:arith
  corr S6}, it is now straightforward to deduce
Theorem~\ref{thm:var(L) asymp united}, our main result.
%
Recall that $\RR_{k}$ are the moments \eqref{eq:R def} of $r_{n}$.
Using the Cauchy-Schwarz inequality on $|r(x)|^{5} = r^2(x)\cdot |r(x)|^3$
together with the bound $\RR_{4} = O\left( \frac{1}{\eigspcdim^{2}}
\right)$ (which follows from the same argument that yielded
\eqref{eq:R6=O(1/N^2)}), and the bound $\RR_{6} (n) = o\left(  \frac{1}{\eigspcdim_{n}^2}  \right)$ from Theorem~\ref{thm:arith corr S6}, we obtain
\begin{equation}
\label{eq:R5=o(1/N^2)}
\RR_{5}(n) = o\left(  \frac{1}{\eigspcdim_{n}^2}  \right).
\end{equation}
Now, using \eqref{eq:var as hat(mu)(4)} together with
\eqref{eq:R5=o(1/N^2)} we obtain (\ref{eq:var(L) asymp united}).

Finally, using the second part of
Proposition~\ref{prop:intermd meas exist} and the definition of
$c_{n}$ (see \eqref{eq:mun def}) we find that any $c \in [1/512,1/256]$
is attainable as a limit.

\section{The $2$-point correlation function of $f_{n}$}

\label{sec:2-point correlation}

In this section we use the Kac-Rice formula~\eqref{eq:kacricevariance} that expresses
$\var(\len_{n})$ as an integral of the (scaled)
$2$-point correlation $K_{2}$ defined in \eqref{eq:kacricemoments}. For this
 we will need to study some aspects of the random field $f_{n}$ first.
From this point on we fix $n$ and will usually suppress the $n$-dependency with no further
note.

\subsection{Joint distribution of values and gradients}
\label{sec:valuesandgradients}
In order to study the variance, we  shall need to study the random vector 
\begin{equation*} 
W = W_{n;x} =(u_{1},u_{2},v_{1},v_{2}) = (f_{n}(0),f_{n}(x),\nabla f_{n}(0),\nabla f_{n}(x)) \in \R^{6}.
\end{equation*}
Since $W$ is a linear transformation of the standard Gaussian $a=(a_{\lambda})\in\R^{\eigspcdim_{n}}$,
its distribution is also centered (or mean zero) Gaussian, and by
the stationarity, $0$ and $x$ may be replaced by any $y$ and $y+x$.


Let
\begin{equation}
\label{eq:D as sum}
D=D_{1\times 2}(x) = \nabla r(x) =
\frac{2\pi i}{\eigspcdim_{n}}\sum\limits_{\lambda\in\Lambda}
 e\left(  \langle \lambda,\, x \rangle  \right) \cdot \lambda
\end{equation}
(cf. \eqref{eq:r def}).
The vector $W$ is centered Gaussian with covariance
matrix (cf. ~\cite{RW}, Section 5.1)
\begin{equation*}
\Sigma = \left( \begin{matrix}
A &B \\ B^{t} &C
\end{matrix}   \right),
\end{equation*}
where
\begin{equation*}
A(x) = \left( \begin{matrix} 1 &r(x) \\ r(x) &1  \end{matrix} \right),
\end{equation*}
\begin{equation*}
B(x) = \left( \begin{matrix}
  0 &D(x) \\ -D(x) & 0
\end{matrix} \right)
\end{equation*}
and
\begin{equation*}
C(x) = \left( \begin{matrix} \frac{E}{2} I_{2} &-H(x)
\\-H(x) &\frac{E}{2}I_{2} \end{matrix} \right),
\end{equation*}
where $H_{2\times 2}(x)$ is the Hessian
\begin{equation}
\label{eq:H summation}
H(x) = \left( \frac{\partial ^{2} r}{\partial x_{i}\partial x_{j}}   \right)
= -\frac{4\pi^2}{\eigspcdim_{n}}\sum\limits_{\lambda\in\Lambda}
e\left(  \langle \lambda,\, x \rangle  \right) \cdot (\lambda^{t}\lambda),
\end{equation}
by \eqref{eq:r def}
(note that $\lambda$ is a row vector so that $\lambda^{t}\lambda$ is a $2\times 2$ matrix).

The covariance matrix of $(\nabla f(0),\nabla f(x))$, conditioned on $f(0)=f(x) = 0$ is
\begin{equation}
\label{eq:tilde{Omega} def}
\tilde{\Omega}_{4\times 4} = C - B^{t}A^{-1}B =
\left( \begin{matrix} \frac{E}{2} I_{2} &-H
\\-H &\frac{E}{2} I_{2} \end{matrix} \right) - \frac{1}{1-r^2}
\left(\begin{matrix}
D^{t}D &rD^{t}D \\ rD^{t}D &D^{t}D
\end{matrix}   \right),
\end{equation}
where we write $r=r(x)$ for brevity. Thus, $K_{2}(x)=\E[\|\tilde{V}_{1}\|\cdot \|\tilde{V}_{2}\|]$ where $\tilde{V}_{i}$ are 2-dimensional random vectors with $(\tilde{V}_{1},\tilde{V}_{2})$ having Gaussian distribution with zero mean and covariance matrix $\tilde{\Omega}(x)$.

\subsection{The scaled $2$-point correlation function}

It is more convenient to work with the {\em scaled} covariance matrix
\begin{equation}
\label{eq:Omega def}
\Omega(x) = \Omega_{n}(x)= \frac{2}{\eigval_{n}} \tilde{\Omega}_{n}(x).
\end{equation}
Then, the scaled $2$-point correlation function defined in \eqref{eq:kacricemoments} may be written as
\begin{equation*}
K_{2}(x) =
\frac{1}{2\pi\sqrt{1-r_{n}(x)^2}} \E[\|V_{1} \| \cdot \| V_{2} \|],
\end{equation*}
where $V_{1},V_{2}$ are centered Gaussians with covariance matrix
$\Omega(x)$.

At the origin $x=0$ the matrix $\Omega(x)$ is singular and hence
corresponds to a covariance matrix of a degenerate Gaussian.  However
for almost all $x\in\T$, $\Omega(x)$ is nonsingular, see
~\cite[Proposition A.1]{RW}.  We claim that $\Omega(x)$ is a small
perturbation of the $4\times 4$ identity matrix $I_{4}$, at least, for
``generic" $x$.  To quantify the latter statement, write
\begin{equation}
\label{eq:Omega=I+(X Y Y X)}
\Omega(x) = I + \left(  \begin{matrix}
X &Y \\ Y &X
\end{matrix} \right),
\end{equation}
where
\begin{equation}
\label{eq:X and Y def}
X = -\frac{2}{\eigval_{n}(1-r^2)}D^{t}D ,\;
Y = -\frac{2}{\eigval_{n}}\left(  H + \frac{r}{1-r^2}D^{t}D \right),
\end{equation}
and both $X=X_{n}(x)$ and $Y=Y_{n}(x)$ are small for ``typical" $x$.

With these computations, we may rewrite the Kac-Rice formula~\eqref{eq:kacricevariance} as follows.
\begin{proposition}[Cf. ~\cite{RW}, Proposition 5.2]
\label{prop:var as K2}
The nodal length variance is given by
\begin{equation}
\label{eq:var as K2}
\var(\len_{n}) = \frac{\eigval_{n}}{2} \int\limits_{\T} \left( K_{2}(x) - \frac{1}{4} \right) dx,
\end{equation}
where $K_{2}$ is the scaled $2$-point correlation function given by
\begin{equation}
\label{eq:K2 as E[|W1| |W2|]}
K_{2}(x) = K_{2;n}(x) = \frac{1}{2\pi\sqrt{1-r_{n}(x)^2}} \E[\|V_{1} \| \cdot \| V_{2} \|];
\end{equation}
here $V_{1},V_{2} \in\R^{2}$ are centered Gaussians with covariance matrix
given by \eqref{eq:Omega=I+(X Y Y X)}, with $X$ and $Y$ as in \eqref{eq:X and Y def}.
\end{proposition}

We shall need the following lemma later.
\begin{lemma}
\label{lem:K2<<1/sqrt(1-r^2)}
The matrices $X_{n}$ and $Y_{n}$ are uniformly bounded (entry-wise), i.e.
\begin{equation}
\label{eq:X,Y=O(1)}
X_{n}(x) , Y_{n}(x) = O(1),
\end{equation}
where the constant involved in the $'O'$-notation is universal.
In particular,
\begin{equation}
\label{eq:K2<<1/sqrt(1-r^2)}
K_{2;n}(x) \ll \frac{1}{\sqrt{1-r_{n}(x)^2}}.
\end{equation}
\end{lemma}

\begin{proof}
To prove that \eqref{eq:X,Y=O(1)} holds it is sufficient to
show that
the diagonal entries of $X$ are uniformly bounded, by \eqref{eq:Omega=I+(X Y Y X)}
(the non-diagonal entries of a covariance matrix are dominated
by the diagonal ones, by the  Cauchy-Schwarz inequality).
For the latter, it is sufficient to notice that the diagonal entries of $\Omega$ are positive,
and the diagonal entries of $X$ are $\le 0$ (recall \eqref{eq:X and Y def}).

To prove that the bound in \eqref{eq:K2<<1/sqrt(1-r^2)} holds, we use
\eqref{eq:K2 as E[|W1| |W2|]}, the Cauchy-Schwarz inequality
and \eqref{eq:X,Y=O(1)}.
\end{proof}

\section{Proof of Proposition \ref{prop:var as hat(mu)(4)}}
\label{sec:proof-arithmetic-prop}

To find the asymptotics of the integral \eqref{eq:var as K2}, we will study
the pointwise asymptotic behaviour of $K_{2}$. Even though we will only be able to determine the precise
asymptotics outside the so-called {\em singular set},
already used in ~\cite{ORW} and ~\cite{RW}, we will prove that the exceptional singular set is {\em small},
so that its contribution is negligible (see Lemma \ref{lem:sing set o(E/N^2)}). To quantify
the last statement, we will control the contribution using a Chebyshev-like inequality, so that
the corresponding error term will naturally involve the moments \eqref{eq:R def} of the covariance function
\eqref{eq:r def}. We improve upon the analysis of ~\cite{ORW} by working with the $6$th moment $\mathcal{R}_{6}(n)$
rather than $\mathcal{R}_{4}(n)$.

\subsection{The singular set}

For $r(x)$ bounded away from $1$ we may expand the $\frac{1}{\sqrt{1-r^2}}$
factor in \eqref{eq:K2 as E[|W1| |W2|]} and
related expressions into the Taylor series around $r=0$.
Since the moments of $r$ are ``small" (by Theorem \ref{thm:arith corr S6}, say),
a Chebyshev-like inequality, implies that
$r(x)$ is small outside a small set. 
This is the main idea behind the notion of the singular set to follow.
We use a slightly stronger definition in order to endow the exceptional set
with a structure as a union of squares, necessary in order to find its contribution
to the integral \eqref{eq:var as K2}.
The following definitions are borrowed directly from ~\cite{ORW}, Section 6.1.

\begin{definition}
A point $x\in \T$ is a positive singular point if there
is a set of frequencies $\Lambda_x\subseteq \Lambda$ with density
$$\frac{|\Lambda_x|}{|\Lambda|}>\frac{7}{8}$$ for which $\cos (2\pi
\langle \lambda, x \rangle) >3/4$ for all $\lambda\in \Lambda_x$.
Similarly we define a negative singular point to be a point $x$
where there is a set $\tilde\Lambda_x\subseteq \Lambda$ of density
$>\frac{7}{8}$ for which $\cos (2\pi\langle \lambda, x\rangle) <-3/4$ for
all $\lambda\in \tilde\Lambda_x$.
\end{definition}

Let $M\approx \sqrt{\eigval_{n}}$ be a large integer.
We decompose the torus $\T$ as a union  of $M^2$
closed squares  $I_{\vec k}$ of side length
$1/M$ centered at $\vec k/M$, $\vec k\in \Z^2$. The squares have disjoint interiors. 

\begin{definition} A square $I_{\vec k}$ is a positive (resp. negative)
singular square if it contains a positive (resp. negative) singular
point.
\end{definition}

\begin{definition}
The singular set $\sing=\sing_{n}$ is the union of all singular squares.
\end{definition}

Note that, by the definition, each singular square contains a singular point; however, points
in $\sing$ are not necessarily all singular. Let $y\in\sing$ be a point lying in a positive singular
cube, $x$ be the corresponding positive singular point lying in the same singular cube and
$\Lambda_{x}\subseteq \Lambda$ the frequency set prescribed by the definition of a positive
singular point. It is easy to see that
\begin{equation*}
\left| \cos (2\pi\langle \lambda, y\rangle) - \cos (2\pi\langle \lambda, x\rangle)\right|
\ll \frac{\sqrt{\eigval_{n}}}{M},
\end{equation*}
where the implied constant is absolute, so that one may choose $M\approx \sqrt{\eigval_{n}}$
for which the latter expression is $\le\frac{1}{4}$; it will then imply
$$\cos (2\pi\langle \lambda, y\rangle) \ge \frac{1}{2}$$
for every $\lambda\in\Lambda_{x}$. We then conclude that
\begin{equation*}
\begin{split}
r(y) &= \frac{1}{|\Lambda|}\sum\limits_{\lambda\in\Lambda_{x}}\cos (2\pi\langle \lambda, y\rangle)
+\frac{1}{|\Lambda|}\sum\limits_{\lambda\in\Lambda\setminus\Lambda_{x}}\cos (2\pi\langle \lambda, y\rangle)
\\&\ge \frac{1}{|\Lambda|}\sum\limits_{\lambda\in\Lambda_{x}}\frac{1}{2} -
\frac{1}{|\Lambda|}\sum\limits_{\lambda\in\Lambda\setminus\Lambda_{x}}1 \ge \frac{7}{16} - \frac{1}{8}
=\frac{5}{16}
\end{split}
\end{equation*}
and, similarly, if $y$ is lying in a negative square then $r(y)\le - \frac{5}{16}$. Hence we
have $|r(y)| \ge \frac{5}{16}$ on all of $\sing$.
We then write
\begin{equation*}
\mathcal{R}_{6}(n) \ge \meas(\sing)\cdot \left(\frac{5}{16} \right)^{4}
\end{equation*}
to obtain the Chebyshev-type inequality
\begin{equation}
\label{eq:meas(sing)=O(R6)}
\meas(\sing) \ll \mathcal{R}_{6}(n).
\end{equation}

It was proven (\cite{ORW}, Section 6.5) that if $S$ is any singular square then
its contribution to the integral \eqref{eq:var as K2} is
\begin{equation*}
\ll \int\limits_{S}|K_{2}(x)| dx \ll \frac{1}{M\sqrt{\eigval_{n}}}.
\end{equation*}
Since the number of the singular cubes is $$\ll M^{2}\meas(\sing),$$ the
total contribution of $\sing$ to \eqref{eq:var as K2} is bounded by
\begin{equation*}
\int\limits_{\sing}|K_{2}(x)| dx \ll M^{2}\meas(\sing) \cdot \frac{1}{M\sqrt{\eigval_{n}}}
=\meas(\sing)\frac{M}{\sqrt{\eigval_{n}}} \ll \mathcal{R}_{6}(n)
\end{equation*}
by \eqref{eq:meas(sing)=O(R6)} and $M\approx \sqrt{\eigval_{n}}$.
The latter is summarized in the following lemma:

\begin{lemma}[Cf. ~\cite{ORW}, Section 6.3]
\label{lem:sing set o(E/N^2)}
The contribution of the singular set to \eqref{eq:var as K2} is bounded by
\begin{equation*}
\int\limits_{\sing} |K_{2}(x)|dx = O\left(\mathcal{R}_{6}(n)\right).
\end{equation*}
\end{lemma}


Lemma \ref{lem:sing set o(E/N^2)} bounds the contribution of the singular set
to the integral in \eqref{eq:var as K2}. The main contribution comes from the nonsingular
set; in order to evaluate it we will need a precise point-wise estimate for
$K_{2}(x)$ in this range; this is given by the following proposition, up to admissible error
terms. (To verify the admissibility, see Lemmas \ref{lem:asymp X,Y} and \ref{lem:asymp r,D,H}).

\begin{proposition}[``Intermediate range" asymptotics for $K_{2}$]
\label{prop:$2$-pnt using covar nonsing}
For $x\in \T\setminus\sing$ we have
\begin{equation}
\label{eq:K2 asymp r,X,Y}
K_{2}(x) = \frac{1}{4}+ L_{2}(x) + \epsilon(x),
\end{equation}
where the main term $L_{2}(x)$ is given by
\begin{equation}
\label{eq:L2 def}
\begin{split}
L_{2}(x) &= \frac{1}{4}\bigg(\frac{r^2}{2}+\frac{\tr{X}}{2}+
\frac{\tr(Y^{2})}{8}+\frac{3}{8}r^{4}-\frac{\tr(XY^2)}{16} -\frac{\tr(X^{2})}{32}
\\&+\frac{\tr(Y^{4})}{256}+\frac{\tr(Y^2)^2}{512}
-\frac{\tr{X}\tr(Y^{2})}{32}+\frac{1}{4}r^2\tr{X}+ \frac{1}{16}r^2\tr(Y^2)
 \bigg)
\end{split}
\end{equation}
with $X=X_{n}(x)$, $Y=Y_{n}(x)$ and $r=r_{n}(x)$,
and the error term $\epsilon(x)$ is bounded by
\begin{equation}
\label{eq:delta error bound}
|\epsilon(x)| =  O\left(r(x)^6+\tr(X^3)+\tr(Y^6)\right)
\end{equation}
\end{proposition}

Proposition~\ref{prop:$2$-pnt using covar nonsing} will be proved in Section~\ref{sec:asymptotics2ptcorelation}.
Assuming it, we arrive at the proof of Proposition~\ref{prop:var as hat(mu)(4)}.
\begin{proof}[Proof of Proposition \ref{prop:var as hat(mu)(4)}
assuming Proposition \ref{prop:$2$-pnt using covar nonsing}]

We invoke Proposition \ref{prop:var as K2} to express the nodal length variance.
Since the contribution of $K_{2}(x)$ to the integral \eqref{eq:var as K2} on $\sing$ is $$O\left(\eigval_{n}\cdot\mathcal{R}_{6}(n) \right)$$
by Lemma \ref{lem:sing set o(E/N^2)}, we have
\begin{equation}
\label{eq:var nonsing + error}
\begin{split}
\var(\len_{n}) &= \frac{\eigval_{n}}{2} \left[ \int\limits_{\nonsing} \left( K_{2}(x) - \frac{1}{4} \right) dx \right]
+ O\left(\eigval_{n}\cdot\mathcal{R}_{6}(n) \right) \\&=
\frac{\eigval_{n}}{2} \int\limits_{\nonsing} L_{2}(x)  dx
+ O\left(\eigval_{n}\cdot \int\limits_{\nonsing} |\epsilon(x)| dx\right) + O\left(\eigval_{n}\cdot\mathcal{R}_{6}(n) \right),
\end{split}
\end{equation}
by Proposition \ref{prop:$2$-pnt using covar nonsing}. Note that
\begin{equation*}
\int_{\nonsing}|\epsilon(x)| dx \le \int_{\T}|\epsilon(x)| dx = O\left(\eigval_{n}\cdot\mathcal{R}_{6}(n) \right),
\end{equation*}
by \eqref{eq:delta error bound} and Lemma \ref{lem:asymp X,Y} to follow
(see parts \ref{it:int(tr(X^3))=o(1/N^2)}-\ref{it:int(tr(Y^6))=o(1/N^2)}),
so that \eqref{eq:var nonsing + error} is
\begin{equation}
\label{eq:var nonsing + error R6}
\var(\len_{n}) =
\frac{\eigval_{n}}{2} \int\limits_{\nonsing} L_{2}(x)  dx
+ O\left(\eigval_{n}\cdot\mathcal{R}_{6}(n) \right).
\end{equation}
We may further note that, since $L_{2}(x)$ is uniformly bounded thanks to Lemma \ref{lem:K2<<1/sqrt(1-r^2)},
\begin{equation*}
\frac{\eigval_{n}}{2} \int\limits_{\sing} L_{2}(x)  dx = O\left( \eigval_{n}\cdot \meas(\sing)  \right) =
O\left(\eigval_{n} \cdot\mathcal{R}_{6}(n)  \right),
\end{equation*}
so that we may rewrite \eqref{eq:var nonsing + error R6} as
\begin{equation}
\label{eq:var as L2 + error R6}
\var(\len_{n}) =
\frac{\eigval_{n}}{2} \int\limits_{\T} L_{2}(x)  dx
+ O\left(\eigval_{n}\cdot\mathcal{R}_{6}(n) \right),
\end{equation}
the upshot being that we are now able to use the definition
\eqref{eq:L2 def} of $L_{2}$ and integrate the RHS of \eqref{eq:L2 def} term by term, as
in Lemma \ref{lem:asymp X,Y} (where the domain of integration is the whole of torus $\T$
rather than $\nonsing$).
We then perform the term-wise integration of \eqref{eq:L2 def} to obtain (with Lemma \ref{lem:asymp X,Y})
\begin{equation*}
\begin{split}
&4\cdot\int\limits_{\T}L_{2}(x)   dx =
\frac{1}{\eigspcdim_{n}} \left(  \frac{1}{2} - \frac{1}{2}\cdot 2 +\frac{1}{8} \cdot 4  \right)
\\&+\frac{1}{\eigspcdim_{n}^2}\bigg( -\frac{1}{2}\cdot 2 - \frac{1}{8} \cdot 4 +
\frac{3}{8}\cdot 3+\frac{1}{16}\cdot 4  - \frac{1}{32}\cdot 8 +
\frac{1}{256}\cdot 2(11+\widehat{\mu}_{n}(4)^{2}) \\&+\frac{1}{512}\cdot 4(7+\widehat{\mu}_{n}(4)^2)
+\frac{1}{32}\cdot 8 -\frac{1}{4}\cdot 2+ \frac{1}{16}\cdot 8
 \bigg) +O\left(\mathcal{R}_{5}(n) \right) \\&= \frac{1}{\eigspcdim_{n}^2} \cdot \frac{1+\widehat{\mu}_{n}(4)^2}{64}+O\left(\mathcal{R}_{5}(n) \right).
\end{split}
\end{equation*}
Collecting all the constants encountered and bearing in mind \eqref{eq:var as L2 + error R6}
yields \eqref{eq:var as hat(mu)(4)}, which is the statement of the present proposition.
\end{proof}

\subsection{Some remarks on arithmetic Berry cancellation}

\label{sec:discussion Berry's cancellation}


While the constant term $\frac{1}{4}$ cancels out with the expectation squared,
the leading nonconstant term of the scaled $2$-point correlation function
(i.e. the leading term of $K_{2}(x)-\frac{1}{4}$) is
\begin{equation*}
\frac{1}{8}\left( r^2+\tr{X} + \frac{\tr(Y^2)}{4} \right) \approx
\frac{1}{8} \left( r^2 - \frac{2}{\eigval_{n}} DD^{t} + \frac{1}{\eigval_{n}^{2}}\tr{H^2}   \right),
\end{equation*}
where we neglected some lower-order terms.
Denote the expression in parenthesis
\begin{equation}
\label{eq:v def}
v(x) := r^2 - \frac{2}{\eigval_{n}} DD^{t} + \frac{1}{\eigval_{n}^{2}}\tr{H^2}.
\end{equation}

We may substitute \eqref{eq:r def}, \eqref{eq:D as sum} and \eqref{eq:H summation} into
\eqref{eq:v def} to rewrite $v(x)$ as
\begin{equation*}
\begin{split}
v(x) = &\frac{1}{\eigspcdim^2}\sum\limits_{\lambda_{1},\lambda_{2}\in \Lambda_{n}}
e\left(  \langle \lambda_{1}+\lambda_{2} , x\rangle  \right)
+ \frac{2}{\eigspcdim^2}\sum\limits_{\lambda_{1},\lambda_{2}\in \Lambda_{n}}
\frac{\lambda_{1} \lambda_{2} ^{t}}{\eigval_{n}/4\pi^2}
e\left( \langle \lambda_{1}+\lambda_{2} , x\rangle \right)
\\&+ \frac{1}{\eigspcdim^2}\sum\limits_{\lambda_{1},\lambda_{2}\in \Lambda_{n}}
\frac{(\lambda_{1} \lambda_{2} ^{t})^2}{(\eigval_{n}/4\pi^2)^2}
e\left(  \langle \lambda_{1}+\lambda_{2} , x\rangle \right)
\\&= \frac{1}{\eigspcdim^2}\sum\limits_{\lambda_{1},\lambda_{2}\in \Lambda_{n}}
\left( 1 + 2\frac{\lambda_{1} \lambda_{2} ^{t}}{n} + \frac{(\lambda_{1} \lambda_{2} ^{t})^2}{n^2}  \right) e\left(  \langle \lambda_{1}+\lambda_{2} , x\rangle \right),
\end{split}
\end{equation*}
on recalling \eqref{eq:eigval=4pi^2 n}.

Note that
\begin{equation*}
 \frac{\lambda_{1} \lambda_{2} ^{t}}{n} = \cos\theta(\lambda_{1},\lambda_{2}),
\end{equation*}
where $\theta(\cdot, \cdot)$ is the angle between two vectors in $\R^{2}$.
Thus we may write, up to lower order terms,
\begin{equation*}
\begin{split}
v(x) &= \frac{1}{\eigspcdim^2}\sum\limits_{\lambda_{1},\lambda_{2}\in \Lambda_{n}}
\left( 1 + 2\cos\theta(\lambda_{1},\lambda_{2}) + \cos\left(\theta(\lambda_{1},\lambda_{2})\right)^2  \right)
e\left(  \langle \lambda_{1}+\lambda_{2} , x\rangle \right)
\\&= \frac{4}{\eigspcdim^2}\sum\limits_{\lambda_{1},\lambda_{2}\in \Lambda_{n}}
\cos\left(\frac{\theta(\lambda_{1},\lambda_{2})}{2}\right)^4
e\left(  \langle \lambda_{1}+\lambda_{2} , x\rangle \right),
\end{split}
\end{equation*}
by the usual trigonometric identities.
Upon integrating \eqref{eq:var as K2}, all the summands vanish except
for $\lambda_{1}+\lambda_{2} = 0$; the corresponding angle $\theta$ is given by
$$\theta=\theta(\lambda_{1},\lambda_{2})=\pi,$$
so that $\cos(\theta/2) = 0$.
Thus the arithmetic cancellation phenomenon in the length variance amounts to
$\cos(\theta/2)^4$ vanishing at $\theta = \pi$.

\subsection{Integrating matrix elements}

We may obtain an asymptotic expression
for the nodal length variance upon using \eqref{eq:var as K2} with Proposition \ref{prop:$2$-pnt using covar nonsing},
provided that we are able to integrate the expressions
on the RHS of \eqref{eq:K2 asymp r,X,Y}, term-wise. This is done in Lemma \ref{lem:asymp X,Y} to follow
immediately. We choose to control the various error terms encountered in terms of the moments of $r$,
$\mathcal{R}_{k}$ (recall the notation \eqref{eq:R def}). It will turn out that we will be able to control the error terms in terms of $\mathcal{R}_{5}$ (and $\mathcal{R}_{6}
\le \mathcal{R}_{5}$), admissible thanks to Theorem \ref{thm:arith corr S6}
via a simple Cauchy-Schwarz argument (see the proof of Theorem \ref{thm:var(L) asymp united}
in Section \ref{sec:proof thm var united}).
The proof of Lemma \ref{lem:asymp X,Y} is left to Section \ref{sec:proof of int mat elem}.

\begin{lemma} As $\eigspcdim_{n}\rightarrow\infty$ we have the
  following estimates.

\label{lem:asymp X,Y}

\vspace{2pt}

\begin{center}
\begin{enumerate}[$1.$]\setlength{\itemsep}{4pt}\setlength{\itemindent}{35pt}

\item
\label{it:tr(X)}
$
\int\limits_{\T} \tr{X(x)}dx = -\frac{2}{\eigspcdim_{n}}-\frac{2}{\eigspcdim_{n}^2}+O\left( \mathcal{R}_{6}(n)  \right).
$

\item
\label{it:tr(Y^2)}
$
\int\limits_{\T} \tr(Y(x)^{2})dx =\frac{4}{\eigspcdim_{n}}-\frac{4}{\eigspcdim_{n}^2} + O\left( \mathcal{R}_{6}(n)  \right).
$

\item
\label{it:tr(XY^2)}
$
\int\limits_{\T} \tr(X(x)Y(x)^2)dx = -\frac{4}{\eigspcdim_{n}^2}+O\left(\mathcal{R}_{5}(n) \right)
$

\item
\label{it:tr(X^2)}
$
\int\limits_{\T} \tr(X(x)^2)dx = \frac{8}{\eigspcdim_{n}^2}+O\left( \mathcal{R}_{6}(n)  \right)
$

\item
\label{it:tr(Y^4)}
$
\int\limits_{\T} \tr(Y(x)^4)dx = \frac{2}{\eigspcdim_{n}^2}(11+\widehat{\mu}_{n}(4)^{2}) +
O\left( \mathcal{R}_{6}(n)  \right).
$

\item
\label{it:tr(Y^2)^2}
$
\int\limits_{\T} \tr(Y(x)^2)^2dx = \frac{4}{\eigspcdim_{n}^2}(3+\widehat{\mu}_{n}(4)^2)+O\left( \mathcal{R}_{6}(n)  \right)
$

\item
\label{it:tr(X)tr(Y^2)}
$
\int\limits_{\T} \tr{X(x)}\tr(Y(x)^2)dx = -\frac{8}{\eigspcdim_{n}^2} + O\left( \mathcal{R}_{6}(n)  \right)
$

\item
\label{it:r^2tr(X)}
$
\int\limits_{\T} r(x)^2\tr{X(x)}dx = -\frac{2}{\eigspcdim_{n}^2}+O\left( \mathcal{R}_{6}(n)  \right)
$

\item
\label{it:r^2tr(Y^2)}
$
\int\limits_{\T} r(x)^2\tr(Y(x)^2)dx = \frac{8}{\eigspcdim_{n}^2} +O\left( \mathcal{R}_{6}(n)  \right).
$

\item
\label{it:int(tr(X^3))=o(1/N^2)}

$
\int\limits_{\T} \tr(X(x)^3)dx =  O(\mathcal{R}_{6}(n)).
$

\item
\label{it:int(tr(Y^6))=o(1/N^2)}

$
\int\limits_{\T} \tr(Y(x)^6)dx =  O\left( \mathcal{R}_{6}(n)     \right).
$

\end{enumerate}
\end{center}

\end{lemma}

\section{Asymptotics for the $2$-point correlation function}
\label{sec:asymptotics2ptcorelation}
The ultimate goal of this section is to prove Proposition \ref{prop:$2$-pnt using covar nonsing}.
To establish the desired asymptotics for \eqref{eq:K2 as E[|W1| |W2|]}, one needs
to understand the behaviour of $\E[\| V_{1} \| \cdot \| V_{2} \|]$ where $(V_{1},V_{2})$ is
a centered Gaussian with covariance $\Omega_{n}$, the latter being a small perturbation of
the indentity matrix, given by \eqref{eq:Omega=I+(X Y Y X)}, where both $X$ and $Y$ are small.
That is, we expand $F(X,Y)=\E[\| V_{1} \| \cdot \| V_{2} \|]$ into a Taylor polynomial of
the entries of $X,Y$, about $X=Y=0$.

The degree of the required Taylor polynomial in each
of the variables is determined according to its (average) order of magnitude and
the admissible error term. In principle, one may compute the polynomial by brute force, computing each
derivative separately, but this approach results in a long and tedious
computation. 
In this manuscript we employ Berry's method
~\cite{Berry 2002} in order
to compute the nodal length fluctuations for the random monochromatic
planar waves.
The following lemma provides the Taylor approximation of $F(X,Y)$
for perturbed standard Gaussian.

\begin{lemma}
\label{lem:E[|W1||W2|] Taylor}

Let $\Delta \in M_{4}(\R)$ be a positive definite matrix such that
\begin{equation*}
\Delta = I+\left(  \begin{matrix}
X &Y \\ Y &X
\end{matrix} \right),
\end{equation*}
where $X,Y\in M_{2}(\R)$ are symmetric, $\mbox{rank}(X)=1$. 
Define
\begin{equation*}
F(X,Y) = \E\left[\|W_{1}\| \cdot \| W_{2} \| \right],
\end{equation*}
where $(W_{1},W_{2})\in\R^{2}\times\R^{2}$ is centered Gaussian with covariance $\Delta$.
Then
\begin{equation*}
\begin{split}
F(X,Y) &=\frac{\pi}{2} \bigg( 1 +\frac{\tr{X}}{2}+
\frac{\tr(Y^{2})}{8}-\frac{\tr(XY^2)}{16} -\frac{\tr(X^{2})}{32}
+\frac{\tr(Y^{4})}{256}\\&+\frac{\tr(Y^2)^2}{512}
-\frac{\tr{X}\tr(Y^{2})}{32}
 \bigg) + O\left(\tr(X^{3})+\tr(Y^{6})\right).
\end{split}
\end{equation*}

\end{lemma}

\begin{proof}[Proof of Proposition
\ref{prop:$2$-pnt using covar nonsing} assuming Lemma \ref{lem:E[|W1||W2|] Taylor}]

Note that since $D^{t}D$ is a rank $1$ matrix, it satisfies
$$\tr(D^{t}D)=DD^{t}.$$
A straightforward application of Lemma \ref{lem:E[|W1||W2|] Taylor} with
$X$ and $Y$ given by \eqref{eq:X and Y def}
yields (for $x\in\nonsing$, $r$ is bounded away from $\pm 1$, so we may write
$\frac{1}{\sqrt{1-r^2}}=1+\frac{1}{2}r^2+\frac{3}{8}r^4+O(r^6)$)
\begin{equation*}
\begin{split}
K_{2}(x) &= \frac{1}{2\pi \sqrt{1-r^2}} \cdot F(X,Y) \\&=
\frac{1}{4\sqrt{1-r^2}} \bigg( 1 +\frac{\tr{X}}{2}+
\frac{\tr(Y^{2})}{8}-\frac{\tr(XY^2)}{16} -\frac{\tr(X^{2})}{32}
\\&+\frac{\tr(Y^{4})}{256}+\frac{\tr(Y^2)^2}{512}
-\frac{\tr{X}\tr(Y^{2})}{32}
 \bigg) +O\left(\tr(X^3)+\tr(Y^6)\right)
\\&= \frac{1}{4}\bigg( 1 +\frac{r^2}{2}+\frac{\tr{X}}{2}+
\frac{\tr(Y^{2})}{8}+\frac{3}{8}r^{4}-\frac{\tr(XY^2)}{16} -\frac{\tr(X^{2})}{32}
\\&+\frac{\tr(Y^{4})}{256}+\frac{\tr(Y^2)^2}{512}
-\frac{\tr{X}\tr(Y^{2})}{32}+\frac{1}{4}r^2\tr{X}+ \frac{1}{16}r^2\tr(Y^2)
 \bigg) \\&+ O\left(r^6+\tr(X^3)+\tr(Y^6)\right).
\end{split}
\end{equation*}
\end{proof}

To present the proof of Lemma \ref{lem:E[|W1||W2|] Taylor} we
need some notation.

\begin{notation}
For a matrix $A$ and a number $a$ we write $A=O(a)$ if the corresponding inequality
holds entry-wise.
\end{notation}

\begin{notation}
\label{not:m(t,s) def}
For $t \in 0$ we denote $m(t):=\min\{ t,1\} $, and for $t,s\in\R$, $$m(t,s):=m(t)\cdot m(s).$$
\end{notation}

\begin{proof}[Proof of Lemma \ref{lem:E[|W1||W2|] Taylor}]

Following Berry, see ~\cite[Eq. (24)]{Berry 2002},
\begin{equation*}
\sqrt{\alpha} = \frac{1}{\sqrt{2\pi}}\int\limits_{0}^{\infty} (1-e^{-\frac{\alpha t}{2}})\frac{dt}{t^{3/2}},
\end{equation*}
we have
\begin{equation}
\label{eq:E[W1E2]=int f}
\E[\|W_{1}\|\cdot \|W_{2}\|] = \frac{1}{2\pi} \iint\limits_{\R_{+}^{2}}
\left[ f(0,0)-f(t,0)-f(0,s)+f(t,s)  \right] \frac{dtds}{(ts)^{3/2}},
\end{equation}
where
\begin{equation}
\label{eq:f(t,s)=1/det(I+M)}
 f^{X,Y}(t,s) =f(x,y):= \E\left[\exp \left( -\frac{1}{2}\left(\|W_{1}\|^{2} + \|W_{2}\|^2\right) \right) \right]
= \frac{1}{\det(I+M)},
\end{equation}
with
\begin{equation*}
M = \left( \begin{matrix}  \sqrt{t}I  & 0\\0 &\sqrt{s}I  \end{matrix}  \right)
\Delta \left( \begin{matrix}  \sqrt{t}I & 0 \\ 0 &\sqrt{s}I  \end{matrix}  \right).
\end{equation*} 
Now by the well-known formula for the determinant of
a block matrix (see e.g. ~\cite{CL}, p. 210), we have
\begin{equation*}
\det(I+M) = \det\left( (1+t)I+tX  \right) \cdot
\det \left((1+s)I+sX-stY((1+t)I+tX)^{-1}Y   \right),
\end{equation*}
so that
\begin{equation}
\begin{split}
\label{eq:block det 1/sqrt(I+M)}
\det(I+M)^{-1/2} &= \det\left( (1+t)I+tX  \right)^{-1/2} \times \\&\times
\det \left((1+s)I+sX-stY((1+t)I+tX)^{-1}Y   \right)^{-1/2}.
\end{split}
\end{equation}

Now we compute each of the two factors of the RHS of \eqref{eq:block det 1/sqrt(I+M)},
up to the admissible error terms $X^{3}$ and $Y^{6}$, as in the formulation of
Lemma \ref{lem:E[|W1||W2|] Taylor}:
\begin{equation}
\label{eq:1/sqrt(det(I+A)) exp}
\det(I+A)^{-1/2} = 1-\frac{1}{2}\tr{A}+\frac{1}{4}\tr({A}^2)+\frac{1}{8}(\tr{A})^2+O(A^3),
\end{equation}
so that the first factor in the RHS of \eqref{eq:block det 1/sqrt(I+M)} is
\begin{equation}
\label{eq:1st factor approx}
\begin{split}
&\det\left( (1+t)I+tX  \right)^{-1/2} = \frac{1}{1+t}  \det \left( I+\frac{t}{1+t}X  \right)^{-1/2}
\\&= \frac{1}{1+t} \cdot \left(  1-\frac{t}{2(1+t)}\tr{X} + \frac{t^2}{4(1+t)^2}\tr(X^2) +
\frac{t^2}{8(1+t)^2}(\tr{X})^2+O(X^3) \right).
\end{split}
\end{equation}
To compute the second factor in the RHS of \eqref{eq:block det 1/sqrt(I+M)} we write
\begin{equation*}
(I+A)^{-1} = I-A+O(A^2),
\end{equation*}
and we then have
\begin{equation}
\label{eq:2nd factor approx}
\begin{split}
&\frac{1}{1+s}\det\left( I+\frac{s}{1+s}X-\frac{st}{(1+s)(1+t)} Y \left(I+\frac{t}{1+t}X\right)^{-1}Y   \right)^{-1/2}
\\&=  \frac{1}{1+s}\det\left( I+\frac{s}{1+s}X-\frac{st}{(1+s)(1+t)}Y^{2}+\frac{st^2}{(1+s)(1+t)^{2}}YXY +
O(YX^2Y)\right)^{-1/2}
\\&= \frac{1}{1+s} \bigg(  1 - \frac{1}{2}\frac{s}{1+s}\tr{X}+\frac{1}{2}\frac{st}{(1+s)(1+t)}\tr(Y^{2})-
\frac{1}{2}\frac{st^2}{(1+s)(1+t)^{2}}\tr(YXY)  \\&+
\frac{1}{4}\frac{s^2}{(1+s)^2}\tr(X^2)+\frac{1}{4}\frac{s^2t^2}{(1+s)^2(1+t)^2}\tr(Y^4)
-\frac{1}{2}\frac{s^2t}{(1+s)^2(1+t)}\tr(XY^2) \\&+
\frac{1}{8}\frac{s^2}{(1+s)^2}\tr(X)^2 + \frac{1}{8}\frac{s^2t^2}{(1+s)^2(1+t)^2} \tr(Y^2)^2
- \frac{1}{4}\frac{s^2t}{(1+s)^2(1+t)}\tr{X}\tr(Y^2) \\&+ O(\tr(X^3)+\tr(Y^6)) \bigg)
\end{split}
\end{equation}
upon using \eqref{eq:1/sqrt(det(I+A)) exp} with $$A=\frac{s}{1+s}X-\frac{st}{(1+s)(1+t)}Y^{2}+\frac{st^2}{(1+s)(1+t)^{2}}YXY +
O(YX^2Y).$$

Cross multiplying \eqref{eq:1st factor approx} and \eqref{eq:2nd factor approx}
and substituting into \eqref{eq:block det 1/sqrt(I+M)} and finally into \eqref{eq:f(t,s)=1/det(I+M)}, we obtain
an asymptotic expression for $f^{X,Y}(t,s)$ of the form
\begin{equation}
\label{eq:f(t,s) approx}
\begin{split}
&f^{X,Y}(t,s) = \frac{1}{(1+t)(1+s)}
\bigg(  1 - \frac{1}{2}\frac{s}{1+s}\tr{X}+\frac{1}{2}\frac{st}{(1+s)(1+t)}\tr(Y^{2})-
\\& \frac{1}{2}\frac{st^2}{(1+s)(1+t)^{2}}\tr(YXY)  +
\frac{1}{4}\frac{s^2}{(1+s)^2}\tr(X^2)+\frac{1}{4}\frac{s^2t^2}{(1+s)^2(1+t)^2}\tr(Y^4)
\\&-\frac{1}{2}\frac{s^2t}{(1+s)^2(1+t)}\tr(XY^2) +
\frac{1}{8}\frac{s^2}{(1+s)^2}(\tr{X})^2 + \frac{1}{8}\frac{s^2t^2}{(1+s)^2(1+t)^2} \tr(Y^2)^2
\\&- \frac{1}{4}\frac{s^2t}{(1+s)^2(1+t)}\tr{X}\tr(Y^2)
-\frac{t}{2(1+t)}\tr{X} + \frac{t^2}{4(1+t)^2}\tr(X^2) +
\frac{t^2}{8(1+t)^2}(\tr{X})^2
\\&+\frac{1}{4} \frac{ts}{(1+t)(1+s)}(\tr{X})^2-\frac{1}{4}\frac{st^2}{(1+s)(1+t)^2}\tr{X}\tr(Y^2)
+ O((\tr(X^3)+\tr(Y^6)) \bigg)
\\ &=
\frac{1}{(1+t)(1+s)}
\bigg(
1 - \frac{1}{2}\left( \frac{s}{1+s}+\frac{t}{1+t} \right)\tr{X}+\frac{1}{2}\frac{st}{(1+s)(1+t)}\tr(Y^{2})\\&-
\frac{1}{2}\left( \frac{st^2}{(1+s)(1+t)^{2}} + \frac{s^{2}t}{(1+s)^2(1+t)}  \right)\tr(XY^2)  \\&+
\left(\frac{3}{8}\frac{s^2}{(1+s)^2} + \frac{3}{8}\frac{t^2}{(1+t)^2} +\frac{1}{4} \frac{ts}{(1+t)(1+s)}\right) \tr(X^2)+\frac{1}{4}\frac{s^2t^2}{(1+s)^2(1+t)^2}\tr(Y^4)
\\& + \frac{1}{8}\frac{s^2t^2}{(1+s)^2(1+t)^2} \tr(Y^2)^2
- \frac{1}{4}\left(\frac{s^2t}{(1+s)^2(1+t)}+\frac{st^2}{(1+s)(1+t)^2} \right) \tr{X}\tr(Y^2)
\\&+ O(\tr(X^3)+\tr(Y^6)) \bigg)
\end{split}
\end{equation}
where we used
\begin{equation}
\label{eq:tr(YXY)=tr(XY^2),tr(X^2)=tr(X)^2}
\tr(YXY)=\tr(XY^2), \;
\tr(X^2)=(\tr{X})^2,
\end{equation}
the latter thanks to $\rk{X}=1$.

It is important to identify \eqref{eq:f(t,s) approx} as the Taylor expansion of $f^{X,Y}(t,s)$ for
fixed $t,s$, as a function of $X,Y$ around $X=Y=0$
(i.e. a Taylor polynomial in terms of the entries of $X$ and $Y$).
In the next step we perform the integration in \eqref{eq:E[W1E2]=int f} by integrating
the various terms in \eqref{eq:f(t,s) approx}. The main problem is that the error term does not depend
on $t,s$, so that its integral against $\frac{1}{t^{3/2}s^{3/2}}$ is divergent at the origin.
We improve the error term in the following way:
define
\begin{equation}
\label{eq:gX,Y def}
g^{X,Y}(t,s):=f(0,0)-f(t,0)-f(0,s)+f(t,s) ,
\end{equation}
so that under the new notation \eqref{eq:E[W1E2]=int f} is
\begin{equation}
\label{eq:E[W1E2]=int g}
\E[\|W_{1}\|\cdot \|W_{2}\|] = \frac{1}{2\pi} \iint\limits_{\R_{+}^{2}} g^{X,Y}(t,s) \frac{dtds}{(ts)^{3/2}}.
\end{equation}
It is evident that for every $X,Y$ the function $g^{X,Y}$ vanishes if either $t=0$ or $s=0$,
so that for $t,s \ge 0$, $$g^{X,Y}(t,s)=O_{X,Y}(ts).$$ We now substitute \eqref{eq:f(t,s) approx} into \eqref{eq:gX,Y def} in order to expand $g^{X,Y}(t,s)$ into a Taylor polynomial around $X=Y=0$; by the latter observation the remainder term may be improved from $O\left(\tr(X^3)+\tr(Y^6)\right)$ to $$O\left(m(t,s)(\tr(X^3)+\tr(Y^6))\right)$$ (recall Notation
\ref{not:m(t,s) def}).
To compute the contribution of each of the summands in \eqref{eq:f(t,s) approx}, we notice that
each of summand splits into a product
$\phi(t)\psi(s)$ for some $\phi$ and $\psi$ (that are read off directly, for example, for the constant term $\frac{1}{(1+t)(1+s)}$,
$\phi(t)=\frac{1}{1+t}$ and $\psi(s) = \frac{1}{1+s}$),
so that the corresponding term of $g^{X,Y}(t,s)$ in \eqref{eq:gX,Y def} is
\begin{equation*}
\phi(t)\psi(s)-\phi(t)\psi(0)-\phi(0)\psi(s)+\phi(0)\psi(0) = (\phi(t)-\phi(0))(\psi(s)-\psi(0)).
\end{equation*}
Therefore, the corresponding term in the integral \eqref{eq:E[W1E2]=int g} splits as well.
We then finally obtain
\begin{equation}
\label{eq: g asymp}
\begin{split}
&g^{X,Y}(t,s) = \frac{ts}{(1+t)(1+s)}
+ \frac{1}{2}\left(\frac{t}{1+t}\frac{s}{(1+s)^2}+ \frac{t}{(1+t)^2}\frac{s}{1+s} \right)\tr{X}\\&+\frac{1}{2}\frac{t}{(1+t)^2}\frac{s}{(1+s)^2}\tr(Y^2)
-\frac{1}{2}\left( \frac{t^2}{(1+t)^3}\frac{s}{(1+s)^2}+
\frac{t}{(1+t)^2}\frac{s^2}{(1+s)^3} \right)\tr(XY^2)
\\&-\left( \frac{3}{8}\frac{t}{1+t}\frac{s^2}{(1+s)^3}
+\frac{3}{8}\frac{t^2}{(1+t)^3}\frac{s}{1+s}-\frac{1}{4}\frac{t}{(1+t)^2}\frac{s}{(1+s)^2} \right)\tr(X^2)
\\&+\frac{1}{4}\frac{t^2}{(1+t)^3}\frac{s^2}{(1+s)^3}\tr(Y^4)+
\frac{1}{8}\frac{t^2}{(1+t)^3}\frac{s^2}{(1+s)^3}\tr(Y^2)^2 \\&-
\frac{1}{4}\left(\frac{t^2}{(1+t)^3}\frac{s}{(1+s)^2}+\frac{t}{(1+t)^2}\frac{s^2}{(1+s)^3} \right)\tr{X}\tr(Y^2)
\\&+ O\left(m(t,s)(\tr(X^3)+\tr(Y^6))\right).
\end{split}
\end{equation}
Note that rather than improving the error term in the last step we may incorporate the
improvement into more precise versions of \eqref{eq:1st factor approx} and \eqref{eq:2nd factor approx}
and then carry the improved error term along;
it would result in the same formula \eqref{eq: g asymp}.

Inserting \eqref{eq: g asymp} into \eqref{eq:E[W1E2]=int g} yields
\begin{equation*}
\begin{split}
\E[\|W_{1}\|\cdot \|W_{2}\|] &= \frac{\pi}{2} \bigg( 1 +\frac{\tr{X}}{2}+
\frac{\tr(Y^{2})}{8}-\frac{\tr(XY^2)}{16} -\frac{\tr(X^{2})}{32}
\\&+\frac{\tr(Y^{4})}{256}+\frac{\tr(Y^2)^2}{512}
-\frac{\tr{X}\tr(Y^{2})}{32}
 \bigg) + O\left(\tr(X^{3})+\tr(Y^{6})\right),
\end{split}
\end{equation*}
using the elementary integrals
\begin{equation*}
\int\limits_{0}^{\infty} \left( \frac{t}{1+t}  \right)\frac{dt}{t^{3/2}} = \pi ;\;
\int\limits_{0}^{\infty} \frac{dt}{(1+t)^2\sqrt{t}} = \frac{\pi}{2};\; \int\limits_{0}^{\infty} \frac{\sqrt{t} dt}{(1+t)^3} = \frac{\pi}{8}.
\end{equation*}
which is the statement of the present lemma.
\end{proof}

\subsection{Proof of Lemma \ref{lem:asymp X,Y}}

\label{sec:proof of int mat elem}

To prove Lemma \ref{lem:asymp X,Y} we will need the following lemma
(which establishes the asymptotics for some expressions involved in $X$ and $Y$  - see \eqref{eq:X and Y def}),
whose proof is relegated to Section \ref{apx:proof asymp r,D,H}.

\begin{lemma}
\label{lem:asymp r,D,H}

We have the following estimates:
\begin{enumerate}[$1.$]\setlength{\itemsep}{10pt}\setlength{\itemindent}{35pt}
\item
\label{it:int r^2,r^4}
$
\int\limits_{\T} r(x)^2 dx = \frac{1}{\eigspcdim_{n}}. \qquad
\int\limits_{\T} r(x)^4 dx = \frac{3}{\eigspcdim_{n}^2}\left(1+O\left(\frac{1}{\eigspcdim_{n}} \right)\right).
$

\item
\label{it:int DD^t, (DD^t)^2}
$
\int\limits_{\T} D(x)D(x)^{t} dx = \frac{\eigval_{n}}{\eigspcdim_{n}}. \qquad
\int\limits_{\T} \left(D(x)D(x)^{t}\right)^2 dx = 2\cdot \frac{\eigval_{n}^2}{\eigspcdim_{n}^2}\left(1+O\left(\frac{1}{\eigspcdim_{n}} \right)\right).
$

\item
\label{it:int r^2DD^t}
$
\int\limits_{\T} r(x)^2D(x)D(x)^{t}dx =  \frac{\eigval_{n}}{\eigspcdim_{n}^2}\left(1+O\left(\frac{1}{\eigspcdim_{n}} \right)\right).
$

\item
\label{it:int tr(H^2), r^2tr(H^2)}
$
\int\limits_{\T} \tr(H(x)^2) dx = \frac{\eigval_{n}^2}{\eigspcdim_{n}}. \qquad \int\limits_{\T} r(x)^2\tr(H(x)^2) dx
= 2\cdot \frac{\eigval_{n}^2}{\eigspcdim_{n}^2}\left(1+O\left(\frac{1}{\eigspcdim_{n}} \right)\right).
$

\item
\label{it:int tr(H^4), tr(H^2)^2}
$
\int\limits_{\T} \tr(H(x)^4) dx = \frac{\eigval_{n}^4}{8\eigspcdim_{n}^2}(11+\widehat{\mu}_{n}(4)^2)
 +O\left(  \frac{\eigval_{n}^4}{\eigspcdim_{n}^3} \right).
 $

$ \qquad \int\limits_{\T} \tr(H(x)^2)^2 dx = \frac{\eigval_{n}^4}{4\eigspcdim_{n}^2} (7+\widehat{\mu}_{n}(4)^2) + O\left( \frac{\eigval_{n}^4}{\eigspcdim_{n}^3}  \right).
$

\item
\label{it:int DD^t tr(H^2)}
$
\int\limits_{\T} D(x)D(x)^t\tr(H(x)^2) dx = \frac{\eigval_{n}^3}{\eigspcdim_{n}^2}
 \left(1+O\left(\frac{1}{\eigspcdim_{n}} \right)\right).
$

\item
$
\int\limits_{\T} r(x)D(x)H(x)D(x)^t dx
= -\frac{1}{2}\cdot \frac{\eigval_{n}^2}{\eigspcdim_{n}^2}\left(1+O\left(\frac{1}{\eigspcdim_{n}} \right)\right).
$

\item
\label{it:int DH^2D^t}
$
\int\limits_{\T} D(x)H(x)^2D(x)^t dx
= \frac{1}{2}\cdot \frac{\eigval_{n}^3}{\eigspcdim_{n}^2}\left(1+O\left(\frac{1}{\eigspcdim_{n}} \right)\right).
$

\item
\label{it:(DD^t)^3 << E^3R6}

$
\int\limits_{\T} (D(x) D(x)^{t})^3 dx = O\left(  \eigval_{n}^{3}\mathcal{R}_{6}(n)   \right).
$

\item
\label{it:r^4DD^t << ER6}

$
\int\limits_{\T} r(x)^{4}D(x)D(x)^{t}   dx = O\left(  \eigval_{n}\mathcal{R}_{6}(n)   \right).
$

\item
\label{it:tr (H^6)^3 << E^6R6}

$
\int\limits_{\T} \tr(H^{6})   dx = O\left(  \eigval_{n}^{6}\mathcal{R}_{6}(n)   \right).
$

\end{enumerate}

\end{lemma}

\begin{proof}[Proof of Lemma \ref{lem:asymp X,Y} assuming Lemma
  \ref{lem:asymp r,D,H}]

In this proof we will suppress the dependence on $x$ (and $n$),
i.e. use the shortcuts $r=r_{n}(x)$, $X=X_{n}(x)$, $Y=Y_{n}(x)$, $D=D_{n}(x)$, $H=H_{n}(x)$.
We have
\begin{equation*}
\int\limits_{\T} \tr{X} dx = \int\limits_{\nonsing}\tr{X} dx + O(\meas(\sing))
\end{equation*}
by the uniform boundedness \eqref{eq:X,Y=O(1)} of $X$. Since $\meas(\sing)$ is
small \eqref{eq:meas(sing)=O(R6)} and on $\nonsing$ we may write
\begin{equation*}
\frac{1}{1-r^{2}} = 1+r^{2}+O(r^{4}),
\end{equation*}
we have
\begin{equation*}
\begin{split}
\int\limits_{\T} \tr{X} dx &= -\frac{2}{\eigval_{n}}\left( \int\limits_{\T} DD^t dx +
\int\limits_{\T} r^2DD^t dx  \right) + O\left( \mathcal{R}_{6}(n)  \right) \\&=
-\frac{2}{\eigspcdim_{n}} - \frac{2}{\eigspcdim_{n}^2} + O\left( \mathcal{R}_{6}(n)  \right)
\end{split}
\end{equation*}
by parts \ref{it:r^4DD^t << ER6}, \ref{it:int DD^t, (DD^t)^2}
and \ref{it:int r^2DD^t} of Lemma \ref{lem:asymp r,D,H}. Arguing in a similar fashion,
we obtain
\begin{equation*}
 \int\limits_{\T} \tr(Y^{2})dx \sim\frac{4}{\eigval_{n}^2}\int\limits_{\T}\left[ \tr(H^2)+ 2rDHD^{t}  \right] dx
= \frac{4}{\eigspcdim_{n}}-\frac{4}{\eigspcdim_{n}^2} + O\left( \mathcal{R}_{6}(n)  \right),
\end{equation*}
\begin{equation*}
\int\limits_{\T} \tr(XY^2)dx \sim -\frac{8}{\eigval_{n}^3} \int\limits_{\T}DH^2D^{t}dx  = -\frac{4}{\eigspcdim_{n}^2}+
O(\mathcal{R}_{5}(n)),
\end{equation*}
\begin{equation*}
\int\limits_{\T} \tr(Y^4)dx \sim \frac{16}{\eigval_{n}^4}\int\limits_{\T} \tr(H^4) + O\left( \mathcal{R}_{6}(n)  \right),
\end{equation*}
\begin{equation*}
\int\limits_{\T} \tr(Y^2)^2dx = \frac{16}{\eigval_{n}^4} \int\limits_{\T}\tr(H^2)^2dx +O\left( \mathcal{R}_{6}(n)  \right)
= \frac{4}{\eigspcdim_{n}^2}(7+\widehat{\mu}_{n}(4)^2)+O\left( \mathcal{R}_{6}(n)  \right).
\end{equation*}
This shows parts \ref{it:tr(X)}, \ref{it:tr(Y^2)},
\ref{it:tr(XY^2)}, \ref{it:tr(Y^4)} and \ref{it:tr(Y^2)^2}, parts \ref{it:tr(X^2)}, \ref{it:tr(X)tr(Y^2)}, \ref{it:r^2tr(X)} and \ref{it:r^2tr(Y^2)} being similar.

To see part \ref{it:int(tr(X^3))=o(1/N^2)}, we notice that, as $X$ is uniformly
bound \eqref{eq:X,Y=O(1)}, and $\meas(\sing)$ is small \eqref{eq:meas(sing)=O(R6)},
it is sufficient to bound
the contribution on $\nonsing$ only, so that
we may assume that $r$ is bounded away from $\pm 1$:
\begin{equation}
\label{eq:int trX^3 << ||D||^3}
\int\limits_{\T} \tr(X^3)dx \ll \frac{1}{E^{3}}\int\limits_{\T}(DD^{t})^3dx +
\RR_{6}(n).
\end{equation}
Part \ref{it:int(tr(X^3))=o(1/N^2)} of Lemma \ref{lem:asymp X,Y} then follows upon applying
part \ref{it:(DD^t)^3 << E^3R6} of Lemma \ref{lem:asymp r,D,H} with \eqref{eq:int trX^3 << ||D||^3}.
The proof for part \ref{it:int(tr(Y^6))=o(1/N^2)} is very similar, using
part \ref{it:tr (H^6)^3 << E^6R6} of Lemma \ref{lem:asymp r,D,H}, and we omit it here.
\end{proof}

\section{Proof of Theorem \ref{thm:arith corr S6}}
\label{sec:proof of S6=o(N^4)}


We begin by recalling some needed results from additive combinatorics.
An additive set is a finite and
non-empty subset of an ambient (additive) abelian group $Z$.
Given an additive set $A$, we define $E(A,A)$, the
{\em additive energy of $A$} by
\begin{equation*}
E(A,A) :=
\left|\left\{ (y_{1},y_{2},y_{3}, y_{4}) \in A^{4} : y_{1}+y_{2} =
y_{3}+y_{4}\right\}\right|.
\end{equation*}
We shall use the following ``large energy version'' of the
Balog-Szemeredi-Gowers theorem (see \cite[Ch.~2.4--5]{Tao Vu}):
\begin{theorem}[BSG]
\label{thm:energy-bsg}
Let $A$ be an additive set,
and let $K \geq 1$. There exists an absolute constant $C$ with the
following property: if
$E(A,A) \geq |A|^{3}/K$, then there exists a subset $A_{1} \subseteq
A$ satisfying
\begin{equation}
  \label{eq:|A1|>K^-C|A|}
|A_{1}| > K^{-C} |A|
\end{equation}
and
\begin{equation}
  \label{eq:|A1+A1|<K^C|A|}
|A_{1}+A_{1}| < K^{C} |A|.
\end{equation}
\end{theorem}
\begin{remark}
Theorem~\ref{thm:energy-bsg} can easily be deduced from
Proposition~2.26 and Theorem~2.31 of
\cite{Tao Vu} as follows: by
Theorem~2.31, $E(A,A) \geq |A|^{3}/K$ implies that there exists
subsets $A_{1} \subseteq A$, $A_{2} \subseteq A$ with $|A_{1}| > K^{-C'}
|A|, |A_{2}| > K^{-C'} |A|$ satisfying $d(A_{1},A_{2}) \leq C' \log K$
(where $d(A_{1},A_{2})$ denotes the Rusza distance between $A_{1},
A_{2}$, and $C'$ is an absolute constant.)  By Proposition~2.26,
$d(A_{1},A_{2}) \leq C' \log K$ implies that $|A_{1}+A_{1}| \leq
K^{C''} |A_1|$ for some $C''$ only depending on $C'$.  Taking $C =
\max(C',C'')$ the result follows.
\end{remark}

If $G$ is a (torsion
free) abelian group, a Generalized Arithmetic Progression (GAP)
of dimension $d$, is a subset $P \subseteq G$
of the form
\begin{equation}
\label{eq:P GAP = sum j xi}
P = \left\{ \xi_0 + \sum_{k=1}^{d} j_{k} \xi_{k} : 0 \leq j_{k} < J_{k}
\text{ for $k=1, \ldots, d$}  \right\}
\end{equation}
with $\xi_0,\ldots, \xi_d \in G$. A GAP $P$ is called {\em proper},
if $|P| = \prod_{k=1}^{d} J_{k}$ (i.e., all elements in the sum $\xi_0 + \sum_{k=1}^{d}
j_{k} \xi_{k}$ are distinct).
It is easy to see that a GAP has ``bounded doubling'', i.e., that
$|A+A|/|A|$ is ``small''.  A surprising converse is Freiman's
celebrated structure theorem --- an additive set with small doubling
is essentially a proper GAP:
\begin{theorem}[\cite{Tao Vu},
  Theorem~5.33]
\label{thm:simple-freiman}
Let $A$ be an additive set in a torsion free group $G$ such that
$|A+A| \leq K|A|$.  Then there exists a proper generalized arithmetic
progression $P$, of rank at most $K-1$, which contains $A$ such
that $$|P| \leq \exp\left(O\left(K^{O(1)}\right)\right) |A|.$$
\end{theorem}

If $A \subset \C$ and $z \in \C$, let $r_{2}(z,A)$ denote the number
of representations of $z$ as a product of two elements from $A$.  The
following result by Chang shows that $r_{2}(z,A)$ is quite small
when $A$ is a GAP.
\begin{proposition}[\cite{chang03-factorization-erdos-szemeredi},
  Proposition~3]
\label{prop:chang}
Let $P \subseteq \C$ be a GAP of the form
\eqref{eq:P GAP = sum j xi}
where $\xi_0,\ldots, \xi_d \in \C$.
Then, for all $z \in \C$,
\begin{equation}
  \label{eq:10}
r_{2}(z,P)  < \exp\left( C_{d} \frac{\log J}{\log\log J}\right)
\end{equation}
where $J = \max_{1 \leq k < d} J_{k}$ and the constant $C_{d}$ only
depends on the dimension $d$ of $P$.
\end{proposition}

\begin{proof}[Proof of Theorem \ref{thm:arith corr S6}]
Assume that $|S_{n}| = o(\eigspcdim_{n}^{4})$ does not hold, i.e., that there
exists some $\delta>0$ such that
\begin{equation}
\label{eq:1}
|S_{6}(n)| > \delta \eigspcdim_{n}^{4}
\end{equation}
for $\eigspcdim_{n}$ arbitrarily large.  Using sum-product type
estimates, we will show that this leads to a contradiction.

To simplify the notation, let $S = S_{6}(n)$, and $N=\eigspcdim_{n}$. From this point
on in this proof we assume that
$\delta$ is fixed; we will write $F \lesssim G$ for some expressions $F$, $G$
(resp. $F \gtrsim G$), if there exists a constant $C$ (which may depend on $\delta$ only),
such that $F\le C \cdot G$ (resp. $F \ge C \cdot G$).

Define
\begin{equation}
\label{eq:An def}
A = A_{n} := (\Lambda_{n}+\Lambda_{n} ) \setminus \{0\};
\end{equation}
note that $A$ then consists of elements that have two (or exactly
one for elements of the form $2 \lambda$, $\lambda \in \Lambda_{n}$)
representations as sums of elements of $\Lambda_{n}$, and we
also note that $A$ is symmetric around the origin. Thus
\begin{equation}
\label{eq:A sim N^2}
|A| = N^{2}/2 + O(N)
\end{equation}
and (\ref{eq:1}) implies
\begin{equation}
\label{eq:y1+y2 in A >= delta N^4}
\left|\{ (y_{1}  , y_{2}) \in A \times A : y_{1}+y_{2} \in A \} \right|
\gtrsim N^{4}.
\end{equation}
(Note that the number of solutions to $\sum_{i=1}^{6} \lambda_i =0$
with the additional constraint that one of $\lambda_{1}+\lambda_{2},
\lambda_{3}+\lambda_{4}, \lambda_{5}+\lambda_{6}$ equals zero is
$O(N^{3})$; this follows immediately on noting that
$\lambda_{i}+\lambda_{j} = z$ has at most four solutions if $z\neq 0$.)

Letting  $\mathds{1}_{A}$ denote
the characteristic function of the set $A\subseteq\Z^{2}$, we have
\begin{equation}
\label{eq:y1+y2 in A = <1A*1A,1A>}
\left| \{ (y_{1}  , y_{2}) \in A \times A : y_{1}+y_{2} \in A \} \right| =
\langle \mathds{1}_{A} \star \mathds{1}_{A}, \mathds{1}_{A} \rangle,
\end{equation}
where we understand both the inner product and the convolution as defined on $L^{2}(\Z^{2})$.
Together with \eqref{eq:y1+y2 in A >= delta N^4} and the Cauchy-Schwarz
inequality the observation \eqref{eq:y1+y2 in A = <1A*1A,1A>} yields
\begin{equation*}
 N^{4} \lesssim \langle \mathds{1}_{A} \star \mathds{1}_{A}, \mathds{1}_{A} \rangle
\le \| \mathds{1}_{A} \star \mathds{1}_{A}  \|_{2} \cdot |A|^{1/2} \sim
\frac{1}{\sqrt{2}}\| \mathds{1}_{A} \star \mathds{1}_{A}  \|_{2} \cdot
N.
\end{equation*}
We may hence estimate the additive energy of $A$ as
\begin{equation*}
  \begin{split}
E(A,A) =
|\{ (y_{1},y_{2},y_{3}, y_{4}) \in A^{4} : y_{1}+y_{2} =
y_{3}+y_{4}\}|
\\
= \| \mathds{1}_{A} \star \mathds{1}_{A} \|_{2}^{2}
\gtrsim N^{6} \gtrsim  |A|^{3}.
  \end{split}
\end{equation*}

We now apply Theorem \ref{thm:energy-bsg} on $A$ with $K=K(\delta)$
constant, to construct a large subset $A_{1}\subseteq A$
having the ``bounded doubling'' property \eqref{eq:|A1+A1|<K^C|A|},
and, in addition
\eqref{eq:|A1|>K^-C|A|}; together with \eqref{eq:A sim N^2} the latter implies
\begin{equation}
\label{eq:|A1|>>N^2}
|A_{1}|  \gtrsim N^{2}.
\end{equation}
Hence, by applying Theorem~\ref{thm:simple-freiman} with $G=\Z^{2}$
and $A_{1} \subset G$,
there exists a proper GAP
\begin{equation}
\label{eq:P GAP = sum j xi concrete}
P = \left\{ \xi_0 + \sum_{k=1}^{d} j_{k} \xi_{k} : 0 \leq j_{k} < J_{k}
\text{ for $k=1, \ldots, d$}  \right\}
\end{equation}
as in \eqref{eq:P GAP = sum j xi}, of bounded dimension (depending on
$\delta$ only),
\begin{equation}
\label{eq:d=d(delta)}
d(\delta) = d(K(\delta)),
\end{equation}
so that $A\subseteq P$ and
\begin{equation*}
|P| \le \exp \left( O\left(K^{O(1)}\right)\right) |A_{1}| \lesssim |A_{1}|.
\end{equation*}

We then have
\begin{equation*}
|A_{1}| = |P \cap A_{1}| \leq |P \cap A| \leq
\sum_{x \in \Lambda_{n}} | (P-x) \cap \Lambda_{n}|,
\end{equation*}
where for the latter inequality we used the definition \eqref{eq:An def} of $A$.
Hence, by \eqref{eq:|A1|>>N^2},
$$
 N^{2} \lesssim
\sum_{x \in \Lambda_{n}} | (P-x) \cap \Lambda_{n}|
$$
and therefore (the length of summation being $N$), $$| (P-x) \cap \Lambda_{n}| \gtrsim N$$ for some $x \in
\Lambda_{n}$.  Replacing $P$ by $P-x$ if necessary, we may assume that
\begin{equation}
\label{eq:P cap Lambdan >> N}
|P \cap \Lambda_{n} | \gtrsim N.
\end{equation}

Using Chang's Proposition \ref{prop:chang} the latter leads to a
contradiction as follows.
If $P = \left\{ \xi_0 + \sum_{k=1}^{d} j_{k} \xi_{k} : 0 \leq j_{k} < J_{k}
\text{ for $k=1, \ldots, d$}  \right\}$, then
$P \cup \overline{P}$ is contained in a GAP, of dimension $2d+1$, of
the form
$$
P' = \left\{\xi_0 + j_{0}(\overline{\xi_{0}}-\xi_0)
+\sum_{k=1}^{d} j_{k} \xi_{k} + \sum_{k=d+1}^{2d} j_{k}
\overline{\xi_{k-d}}      \right\}
$$
where $0 \leq j_{0} < 2$, $0 \leq j_{k} < J_{k}$ for $k=1,\ldots,d$, and
$0 \leq j_{k} < J_{k-d}$ for $k=d+1,\ldots,2d$.

Considering $P'$ as a subset of $\Z + i\Z$ it is clear (since for
every $z\in \Lambda_{n}$, $z \cdot \overline{z} = n$) that
\begin{equation}
\label{eq:r2(n,P') >> N}
r_{2}(n,P') \geq |P' \cap \Lambda_{n}  | \gtrsim N
\end{equation}
by \eqref{eq:P cap Lambdan >> N}.
On the other hand, Proposition \ref{prop:chang} applied on $P'$ implies that
\begin{equation*}
r_{2}(n,P') < \exp\left(C_{2d+1} \frac{\log J}{\log\log{J}}\right)
\end{equation*}
where $$J= \max_{1\leq k \leq d} J_{k} \leq |P| \lesssim N^{2},$$
with $J_{k}$ as in \eqref{eq:P GAP = sum j xi concrete}, and thus
\begin{equation*}
r_{2}(n,P') < \exp\left(C_{2d+1} \frac{\log N}{\log\log{N}}\right).
\end{equation*}
Combined with \eqref{eq:r2(n,P') >> N} the latter estimate implies
\begin{equation*}
N \lesssim \exp\left(C_{2d+1} \frac{\log N}{\log\log{N}}\right),
\end{equation*}
or, taking logarithm of both sides,
\begin{equation*}
\log N \le C \frac{\log N}{\log\log{N}}
\end{equation*}
for some $C=C(\delta)$ that may depend on $\delta$ only (by \eqref{eq:d=d(delta)}).
This is clearly impossible for $N$ arbitrarily large,
and the desired contradiction concludes the proof.
\end{proof}


\section{Probability measures on $\Scircle$ arising
 from $\Lambda_n$}
\label{sec:prob-meas-scircle}

Recalling that $S = \{ n\in\Z :\: n=a^2+b^2,\, a,b\in\Z \}$,
define
\begin{equation*}
S(x) := \{ n\in S:\: n\le x \},
\end{equation*}
and for a subset $S'\subseteq S$ similarly define
$S'(x) := \{ n\in S':\: n\le x \}.$
We say that a set $S'\subseteq S$ has {\em asymptotic density} $ s \in
[0,1]$ if
$ 
\lim\limits_{x\rightarrow\infty} \frac{S'(x)}{S(x)} = s.
$ 
Further, we say that a subsequence $(n_{i})_{i\geq 1}$ of elements in
$S$ is {\em thin} if the subset  $\{n_{i}\}_{i\geq 1} \subset S$ has
asymptotic density zero.

It is known ~\cite{Landau}, that as $x\rightarrow\infty$,
$
S(x) \sim \frac{c x}{\sqrt{\log{x}}}
$  where $c>0$ is known as the Landau-Ramanujan constant;
in particular $\eigspcdim_{n}$ grows as $\sim c\sqrt{\log{x}}$ on
average for $n\le x$.  Moreover, a straightforward modification of an
Erd\"os-Kac type argument to the set $S$ shows that
$$
|\{ n \in S(x) : \log \eigspcdim_{n} \gg \log \log n \}|
= |S(x)| \cdot (1+o(1))
$$
as $x \to \infty$, and consequently there exists a density one subset
$S' \subset S$ such that $\eigspcdim_{n} \to \infty$ if $n \in S'$ and
$n \to \infty$.

Further, the lattice points $\Lambda_{n}$ are
equidistributed on $\Scircle$ along {\em generic} subsequences of energy
levels (see e.g. ~\cite[Proposition 6]{FKW} in the
following sense\footnote{Proposition 6 in ~\cite{FKW}
  implies that all the exponential sums are $o(1)$ for a density one
  sequence of energy levels. The equidistribution follows from the
  Weyl's criterion.}: there
exists a density $1$ subsequence $S''\subseteq S$ so that $\mu_{n}
\underset{n\in S''}{\Rightarrow} \nu$, where $\nu$ is the uniform
probability measure $d\nu(\theta) = \frac{1}{2\pi}d\theta$ on
$\Scircle \cong \R/2\pi\Z$.
(As usual, the notation $\upsilon_{i} \Rightarrow \upsilon$ stands for
weak convergence of probability measures on $\Scircle$, i.e., that
$\int fd\upsilon_{i}\rightarrow \int f d\upsilon$ for every continuous
bounded test function $f$.)
In particular, for a generic sequence of elements $n \in S$, $\eigspcdim_{n}
\to \infty$ and the points in  $\Lambda_n$ are
equidistributed in $\Scircle$.

In the other direction, Cilleruelo ~\cite{Cil} has shown that there
are {thin}  sequences $\left( \eigval_{n_{i}} \right)_{i \geq 1}$
with $\eigspcdim_{n_{i}}\rightarrow\infty$, such that
$\mu_{n_{i}}$ converges to the atomic probability measure supported at
the $4$ symmetric points $\pm 1$, $\pm i$:
\begin{equation}
\label{eq:mu ni -> atomic Cil}
\mu_{n_{i}} \Rightarrow
\nu_{0}:=\frac{1}{4}\sum\limits_{k=0}^{3}\delta_{i^{k}}.
\end{equation}

%
%


\subsection{Some number theoretic prerequisites on Gaussian integers}

Before proceeding with the proof of Proposition \ref{prop:intermd meas
  exist} we begin with some number theoretic preliminaries on
$\mu_{n}$ (see e.g. ~\cite{Cil}).

To describe $\mu_{n}$ we recall some basic facts about
Gaussian integers.  Given a prime $p \equiv 1 \mod 4$, the equation $x^{2}+y^{2} =
p$ has exactly eight solutions in integers $x,y$, and there is a
unique solution satisfying $0 \leq y_{p} \leq x_{p}$.  We can hence
attach an angle $\theta_p \in [0,\pi/4]$ to each such $p$ by writing
$x_{p}+iy_{p} = \sqrt{p} e^{i \theta_{p}}$.  On the other hand, given
a prime $q \equiv 3 \mod 4$, the equation $x^2+y^2 = q$ has no
solutions; whereas $x^{2}+y^{2}=2$ has exactly four solutions.
Moreover, the following holds for the ring of Gaussian integers: the
units are given by $i^{k}$ for $k \in \{0,1,2,3\}$, the set of
Gaussian primes are, up to units, given by $1+i$, primes $q \in
\Z^{+}$ with $q \equiv 3 \mod 4$, and to each prime $p \equiv 1 \mod
4$ there corresponds two Gaussian primes, namely $x_{p}+iy_{p} =
\sqrt{p} e^{i \theta_{p}}$ and $x_{p}-iy_{p} = \sqrt{p} e^{-i\theta_{p}}$.

The elements of $\Lambda_{n}$ can then be parametrized as follows: let
$$
n = 2^{e_{2}} \cdot \prod_{p^{e_p}||n} p_i^{e_p}
\cdot \prod_{q_i^{2e_{q}}||n} q_i^{2e_q}
$$
where $p_{i}$ and $q_{i}$ are all the primes satisfying $p_{i} \equiv 1 \mod 4$
and $q \equiv 3 \mod 4$.  Each pair $(x,y)$ arises as follows: with $z =
x+iy$, we have
$$
z = x+iy = i^{k} \cdot (1+i)^{e_{2}} \cdot
\prod_{p^{e_p}||n} (\sqrt{p}^{e_{p}} e^{i (e_{p}-2l_{p})\theta_{p}})
\cdot \prod_{q_i^{2e_{q}}||n} q_i^{e_q}
$$
where $k \in \{0,1,2,3\}$, and $0 \leq l_{p} \leq e_{p}$ for each
$p|n$.

We can now describe $\mu_n$ as convolutions over prime powers:
define $$\mu_1 := \frac{1}{4}\sum_{k=0}^{3} \delta_{i^{k}}$$ ($\mu_{1} = \nu_{0}$
as in \eqref{eq:mu ni -> atomic Cil}),
$\tilde{\mu}_{2^{e_2}} :=
\delta_{((1+i)/\sqrt{2})^{e_{2}}}$, and
$$
\tilde{\mu}_{p^{e_p}} := \frac{1}{e_{p}+1}
\sum_{l_{p}=0}^{e_{p}} \delta_{ e^{i  (e_{p}-2l_{p}) \theta_p}},
$$
(the ``desymmetrized" version of $\mu_{n}$).
Then
$$
\mu_n  = \mu_1 \star (\star_{p|n} \tilde{\mu}_{p^{e_p}})
$$
where the convolution of two measures $\mu,\mu'$ on $\Scircle$ is given by
$$(\mu \star \mu')(z) = \int_{S^{1}} \mu(w) \mu'(z/w) dw.$$

\subsection{Proof of Proposition \ref{prop:intermd meas exist}}

\begin{proof}
That $\nu_{0}$ and $\nu_{\pi/4}$ arise as weak limits of
$(\mu_{n_{i}})_{i\geq 1}$ was already noted in the introduction of
Section~\ref{sec:prob-meas-scircle}.

To show that the same is true for $\nu_a$ for $a \in (0,\pi/4)$ we
argue as follows:
An easy consequence of Gaussian primes being equidistributed in
sectors (cf. \cite{hecke20-eine-neue-art-I,hecke20-eine-neue-art-II})
is that there exists an infinite
sequence of primes $p_{1} < p_{2} < \ldots$ such that $\theta_{p_j}
\to 0$, where each $p_{j} \equiv 1 \mod 4$ (also see~\cite{Cil}.)
To proceed we construct a sequence of integers $n_{j}$ such that
\begin{equation}
\label{eq:munk->upsa}
\tilde{\mu}_{n_{k}} \to \tilde{\nu}_{a};
\end{equation}
this will immediately  imply that
$
\mu_{n_{k}} \to \nu_{a}
$ 
since $\mu_{n_{k}}$ and $\nu_{a}$ are the symmetrized versions of
$\tilde{\mu}_{n_{k}}$ and $\tilde{\nu}_{a}$ respectively.

Thus, let $e_{k} = [a/\theta_{p_{k}}]$, where $[a/\theta_{p_{k}}]$ denotes
the integer part of $a/\theta_{p_{k}}$, and define
$$
n_{k} = p_{k}^{e_{k}}.
$$
Then, for $f$ any continuous function on $S^{1}$,
\begin{equation}
\label{eq:int(f(t))dmunk}
\int_{S^{1}} f(\theta) d\tilde{\mu}_{n_k}(\theta) =
\frac{1}{e_{k}+1}
\sum_{l=0}^{e_{k}}  f(e^{i (e_{k}-2l) \theta_{p_{k}}}).
\end{equation}
The latter is the $2\theta_{p_{k}}$-spaced Riemann sum for the integral
$$
\frac{1}{2\alpha_{k}}\int\limits_{-\alpha_{k}}^{\alpha_{k}}f(e^{i\theta})d\theta
$$
with $$\alpha_{k} = \theta_{p_{k}}\cdot [a/\theta_{p_k}]  =  a- \theta_{p_{k}} \cdot \{ a/\theta_{p_{k}} \},$$
where $\{\cdot\}$ is the fractional part of a real number.

Note that since $\theta_{p_{k}}\rightarrow 0$ (so that the Riemann sum spacing vanishes),
$$\left |\theta_{p_{k}} \cdot \{ a/\theta_{p_{k}} \} \right|  \le   \theta_{p_{k}} \rightarrow 0, $$
and so $\alpha_{k}\rightarrow a$. Therefore, as $k\rightarrow\infty$,
we have
\begin{equation*}
\int\limits_{S^{1}} f(\theta) d\tilde{\mu}_{n_k}(\theta)
=
\frac{1}{e_{k}+1}
\sum_{l=0}^{e_{k}}  f(e^{i (e_{k}-2l) \theta_{p_{k}}})
\rightarrow
\frac{1}{2a}\int\limits_{-a}^{a}f(e^{i\theta})d\theta
\end{equation*}
for any continuous test function $f$.  Thus all $\nu_a$ are
attainable as limiting measures, and the proof of the first statement is concluded. 

To see that the Fourier coefficent $\widehat{\nu}_{a}(4)$, for $a\in
[0,\frac{\pi}{4}]$,  attains all values in $[0,1]$ it is
sufficient to notice that
$$
\widehat{ \nu}_{0}(4) = \widehat{ \mu_{1}}(4) = 1, \quad
\widehat{\nu_{\pi/4}}(4)  = 0,
$$
and clearly the function $a \mapsto \widehat{ \nu}_{a}(4)$ is
continuous. Therefore, by the intermediate value theorem, given any
value $b \in [0,1]$, there exists a number $a=a(b)$ so that
$\nu_{a}(4) = b$.  Since $\nu_{a}$ is attainable for all $a$, the
second statement of Proposition \ref{prop:intermd meas exist} follows.
\end{proof}

\section{Proof of Lemma \ref{lem:asymp r,D,H}}
\label{apx:proof asymp r,D,H}

For a probability measure $\mu$ on $\Scircle$ we define
\begin{equation*}
B_{4}(\mu) :=\int\limits_{z_{1},z_{2}\in\Scircle}
\langle z_{1}, z_{2} \rangle  ^{4} d\mu(z_{1})d\mu(z_{2}),
\end{equation*}
the $4$th moment of cosine of the angle between two random points on $\Scircle$
drawn independently according to $\mu$.
For instance, if $\mu=\mu_{n}$ are the atomic measures in \eqref{eq:mun def},
\begin{equation}
\label{eq:B4(n) def}
B_{4}(\mu_{n})=\frac{1}{\eigspcdim_{n}^2 n^{4}}\sum\limits_{\lambda_{1},\lambda_{2} \in \Lambda_{n} } \langle \lambda_{1},\, \lambda_{2} \rangle ^4 =
\frac{1}{\eigspcdim_{n}^2}\sum\limits_{\lambda_{1},\lambda_{2} \in \Lambda_{n} }
\cos(\theta(\lambda_{1},\lambda_{2}))^{4}
\end{equation}
is the $4$th moment of cosine of the angle $\theta(\lambda_{1},\lambda_{2})$
between two random uniformly and independently drawn $\Lambda_{n}$-points $\lambda_{1},\lambda_{2}$. While the
expression $B_{4}(n)$ comes up naturally from some of the expressions
evaluated in Lemma \ref{lem:asymp r,D,H}, it is simply related to $\widehat{\mu}_{n}(4)$
as in the following lemma.

\begin{lemma}
\label{lem:B4=3/4+1/8hat(mu)(4)2}
For any probability measure $\mu$ on $\Scircle$, invariant w.r.t. $x\mapsto ix$ and $x\mapsto \bar{x}$,
we have
\begin{equation*}
B_{4}(\mu) = \frac{3}{8}+\frac{1}{8}\widehat{\mu}(4)^{2}.
\end{equation*}
\end{lemma}

\begin{proof}
Use for $z_{1},z_{2} \in \Scircle$
\begin{equation*}
\langle z_{1},z_{2} \rangle = \frac{z_{1}\bar{z_{2}}+\bar{z_{1}}z_{2}}{2}
\end{equation*}
together with the binomial formula for $\langle z_{1},z_{2} \rangle^{4}$ or, alternatively,
the standard identity
$$\cos(\theta)^{4} = \frac{3}{8}+\frac{1}{8}\cos(4\theta)+ \frac{1}{2}\cos(2\theta)$$
to rewrite $B_{4}(\mu)$ as
\begin{equation*}
B_{4}(\mu) = \frac{3}{8}+\int\limits_{z_{1},z_{2}\in\Scircle}
 \left( \frac{1}{8}\Re\left(z_{1}^4\bar{z_{2}}^4\right) + \frac{1}{2}\Re\left(z_{1}^2\bar{z_{2}}^{2}\right)  \right)d\mu(z_{1})d\mu(z_{2}).
\end{equation*}
The statement of the present lemma follows upon noting that
$ 
\int\limits_{\mathcal{S}} z^{4}d\mu(z) = \int\limits_{\mathcal{S}} \bar{z}^{4}d\mu(z) = \widehat{\mu}(4) \in \R
$ 
and
$ 
\int\limits_{\mathcal{S}} z^{2}d\mu(z) = \int\limits_{\mathcal{S}}
\bar{z}^{2}d\mu(z)
$ 
vanish by the symmetry assumptions.
\end{proof}

\begin{proof}[Proof of Lemma \ref{lem:asymp r,D,H}]

In order to evaluate the integrals we will use \eqref{eq:r def}, \eqref{eq:D as sum} and
\eqref{eq:H summation}, and the orthogonality relations of the exponentials
\begin{equation}
\label{eq:orthog rel exp}
\int\limits_{\T} e\left( \langle \lambda,\, x  \rangle  \right)dx
= \begin{cases}
1 &\lambda = 0, \\ 0 &\text{otherwise.}
\end{cases}
\end{equation}
Most of the computations are similar in nature, and we will only show a few examples in detail,
omitting the rest.

The statement of part \ref{it:int r^2,r^4} of the present lemma concerning the second moment of $r$
is evident in light of \eqref{eq:r def} and \eqref{eq:orthog rel exp}. For the $4$th moment,
we have
\begin{equation}
\label{eq:int r^4=1/N^4 |S4|}
\int\limits_{\T} r(x)^{4} dx = \frac{1}{\eigspcdim_{n}^{4}} | S_{4}(n)|,
\end{equation}
where
\begin{equation*}
S_{4}(n) = \left\{(\lambda_{1},\ldots \lambda_{4})\in \Lambda_{n} : \: \sum\limits_{i=1}^{4}\lambda_{i} = 0 \right\}
\end{equation*}
is the length-$4$ correlation set of frequencies (cf. \eqref{eq:S6 def}).
Note that
since two circles may have at most $2$ intersections (i.e. circles of radius
$\sqrt{n}$ centered at $0$ and $\lambda_{1}+\lambda_{2}$),
$(\lambda_{1},\ldots, \lambda_{4})\in S_{4}(n)$
implies that either of the following holds:
\begin{equation}
\label{eq:fine structure of S4}
\begin{split}
(\lambda_{1}=-\lambda_{2} \text{ and } \lambda_{3}=-\lambda_{4}) \text{ or }\\
(\lambda_{1}=-\lambda_{3} \text{ and } \lambda_{2}=-\lambda_{4}) \text{ or } \\
(\lambda_{1}=-\lambda_{4} \text{ and } \lambda_{2}=-\lambda_{3}).
\end{split}
\end{equation}
Conversely, every tuple of either of the forms above is lying inside $S_{4}(4)$.
In particular,
\begin{equation*}
| S_{4}(n)| = 3\eigspcdim_{n}^{2}\left(1+O\left( \frac{1}{\eigspcdim_{n}}   \right)   \right),
\end{equation*}
the error term being an artifact of the existence of degenerate tuples
of the form
$$(\pm\lambda,\pm\lambda,\pm\lambda,\pm\lambda)\in S_{4}(4)$$
(with precisely two plus and two minus signs).
Part \ref{it:int r^2,r^4} of the present lemma then follows upon substituting the latter into
\eqref{eq:int r^4=1/N^4 |S4|}. We will use the fine structure \eqref{eq:fine structure of S4}
of $S_{4}(n)$ in the course of the proof of most of the other statements of the present lemma.

Now we turn to part \ref{it:int DD^t, (DD^t)^2} of the present lemma.
While the first statement is clear from \eqref{eq:D as sum} and \eqref{eq:orthog rel exp},
to show the other statement, we invoke the fine structure \eqref{eq:fine structure of S4} of $S_{4}(n)$.
We have
\begin{equation*}
\begin{split}
&\int\limits_{\T} (D(x)D(x)^{t})^2 dx =
\frac{(2\pi)^4}{\eigspcdim_{n}^{4}}\sum\limits_{(\lambda_{1},\ldots,\lambda_{4})\in S_{4}(n)} \lambda_{1}\lambda_{2}^{t}\lambda_{3}\lambda_{4}^{t}
\\&= \frac{(2\pi)^4}{\eigspcdim_{n}^{4}}\left[\eigspcdim_{n}^2n^{2} + \sum\limits_{\lambda_{1},\lambda_{2}\in \Lambda_{n}} \langle \lambda_{1}, \lambda_{2} \rangle ^{2} +
\sum\limits_{\lambda_{1},\lambda_{2}\in \Lambda_{n}} \lambda_{1}\lambda_{2}^{t}\lambda_{2}\lambda_{1} +
O\left( \eigspcdim n^2  \right) \right].
\end{split}
\end{equation*}
The result of the present computation then follows upon making the simple observations
\begin{equation*}
\sum\limits_{\lambda\in \Lambda_{n}} \langle \lambda , \xi \rangle ^{2} = \frac{1}{2} \eigspcdim_{n} n
\| \xi \|^{2}
\end{equation*}
for every $\xi\in\R^{2}$ (see ~\cite{ORW}, Lemma 5.2) and
\begin{equation*}
\sum\limits_{\lambda\in\Lambda_{n}} \lambda^{t}\lambda = \frac{1}{2} \eigspcdim_{n} n \cdot I_{2}.
\end{equation*}

To show part \ref{it:int r^2DD^t} we note
\begin{equation*}
\int\limits_{\T} r(x)^2D(x)D(x)^{t}dx = -\frac{(2\pi)^{2}}{\eigspcdim^{4}}
\sum\limits_{(\lambda_{1},\ldots,\lambda_{4})\in S_{4}(n)} \lambda_{3}\lambda_{4}^{t},
\end{equation*}
and only tuples with $\lambda_{3}=-\lambda_{4}$ contribute to the latter summation,
i.e. those of the first type in \eqref{eq:fine structure of S4}.
The computation for part \ref{it:int tr(H^2), r^2tr(H^2)} is very similar
to what we encountered before, using \eqref{eq:H summation} with \eqref{eq:orthog rel exp}
for the first statement, and exploiting the fine structure \eqref{eq:fine structure of S4}
of $S_{4}(n)$ for the second one
\begin{equation*}
\int\limits_{\T} \tr(H(x)^2) dx = \frac{(4\pi^{2})^{2}}{\eigspcdim^{2}}
\sum\limits_{(\lambda_{1},\ldots,\lambda_{4})\in S_{4}(n)} \lambda_{3}^{t}\lambda_{3}\lambda_{4}^{t}
\lambda_{4}.
\end{equation*}

To compute the integrals in part \ref{it:int tr(H^4), tr(H^2)^2} we exploit Lemma
\ref{lem:B4=3/4+1/8hat(mu)(4)2}. Similarly to the previous computations, we have by
\eqref{eq:H summation} and \eqref{eq:orthog rel exp}
\begin{equation*}
\begin{split}
&\int\limits_{\T} \tr(H(x)^4) dx = \frac{(4\pi^{2})^{4}}{\eigspcdim_{n}^{4}}
\sum\limits_{(\lambda_{1},\ldots,\lambda_{4})\in S_{4}(n)}
\tr(\lambda_{1}^{t}\lambda_{1}\lambda_{2}^{t}\lambda_{2}\lambda_{3}^{t}\lambda_{3}\lambda_{4}^{t}\lambda_{4})
\\&= \frac{(4\pi^{2})^{4}}{\eigspcdim_{n}^{4}} \bigg[
\sum\limits_{\lambda_{1},\lambda_{2}\in\Lambda_{n}} \tr\left(\lambda_{1}^{t}\lambda_{1}\lambda_{1}^{t}\lambda_{1}\lambda_{2}^{t}\lambda_{2}\lambda_{2}^{t}\lambda_{2}\right)
+ \sum\limits_{\lambda_{1},\lambda_{2}\in\Lambda_{n}} \tr\left(\lambda_{1}^{t}\lambda_{1}\lambda_{2}^{t}\lambda_{2}\lambda_{1}^{t}\lambda_{1}\lambda_{2}^{t}\lambda_{2}\right)  + \\& \sum\limits_{\lambda_{1},\lambda_{2}\in\Lambda_{n}} \tr\left(\lambda_{1}^{t}\lambda_{1}\lambda_{2}^{t}\lambda_{2}\lambda_{2}^{t}\lambda_{2}\lambda_{1}^{t}\lambda_{1}\right)   \bigg] + O\left(  \frac{\eigval_{n}^{4}}{\eigspcdim_{n}^{3}} \right) \\&=
\frac{(4\pi^{2})^{4}}{\eigspcdim_{n}^{4}} \bigg[
n^2\sum\limits_{\lambda_{1},\lambda_{2}\in\Lambda_{n}} \langle \lambda_{1},\lambda_{2}\rangle ^2
+ \sum\limits_{\lambda_{1},\lambda_{2}\in\Lambda_{n}} \langle \lambda_{1},\lambda_{2}\rangle ^4
+n^2\sum\limits_{\lambda_{1},\lambda_{2}\in\Lambda_{n}} \langle \lambda_{1},\lambda_{2}\rangle ^2
\bigg] + O\left(  \frac{\eigval_{n}^{4}}{\eigspcdim_{n}^{3}} \right)
\\&= \frac{\eigval_{n}^{4}}{\eigspcdim_{n}^{2}} \left( 1+B_{4}(n)  \right)+ O\left(  \frac{\eigval_{n}^{4}}{\eigspcdim_{n}^{3}} \right) ,
\end{split}
\end{equation*}
where we used sums as above and the definition \eqref{eq:B4(n) def} of $B_{4}(n)$.
Using Lemma \ref{lem:B4=3/4+1/8hat(mu)(4)2}
we may then rewrite the latter expression in terms of $\widehat{\mu}_{n}(4)$, as in
the statement of the present lemma.

A similar computation shows that the second integral in part \ref{it:int tr(H^4), tr(H^2)^2}
is given by
\begin{equation*}
\int\limits_{\T} \tr(H(x)^2)^2 dx = \frac{\eigval_{n}^{4}}{\eigspcdim_{n}^{2}} \left( 1+2B_{4}(n)  \right)+ O\left(  \frac{\eigval_{n}^{4}}{\eigspcdim_{n}^{3}} \right)
\end{equation*}
and using Lemma \ref{lem:B4=3/4+1/8hat(mu)(4)2} again yields the result given.
Evaluating the integrals for parts \ref{it:int DD^t tr(H^2)}-\ref{it:int DH^2D^t} of
the present lemma is straightforward and very similar to the above computations,
and we omit it here.

We now prove part \ref{it:(DD^t)^3 << E^3R6} of the present lemma.
We have by symmetry
\begin{equation}
\label{eq:int ||D||^3 << int der^6}
\int\limits_{\T}(D(x)D(x)^t)^{3}dx = \int\limits_{\T}\| D(x) \|^6 dx \ll
\int\limits_{\T} \left( \frac{\partial r}{\partial x_{1}}  \right)^6 dx,
\end{equation}
whence
\begin{equation}
\label{eq:int der^6<<E^3int r^3}
\begin{split}
&\int\limits_{\T} \left( \frac{\partial r}{\partial x_{1}}  \right)^6 dx =
\frac{(2\pi)^6}{\eigspcdim_{n}^6}\int\limits_{\T} \left( \sum\limits_{\lambda=(\lambda^{1},\lambda^{2})\in\Lambda} \lambda^{1}
\sin(2\pi  \langle\lambda, x \rangle)  \right)^6 dx  \\&=
\frac{(2\pi)^6}{\eigspcdim_{n}^6}\sum\limits_{\substack{\lambda_{1},\ldots,\lambda_{6}\in\Lambda
\\ \sum\limits_{i=1}^{6}\lambda_{i} = 0}} \lambda_{1}^{1}\cdot \ldots\cdot \lambda_{6}^{1}
 \ll \frac{E^{3}}{\eigspcdim_{n}^6} \cdot |S_{6}(n)|
= E^{3}\int\limits_{\T}r(x)^6dx,
\end{split}
\end{equation}
by \eqref{eq:int r^6=1/N^6 S6(n)}, where we used the uniform bound
$$|\lambda_{i}^{1}| \le \sqrt{n} \ll \sqrt{E}. $$ The present statement of Lemma
\ref{lem:asymp r,D,H} then follows upon substituting
\eqref{eq:int der^6<<E^3int r^3} into \eqref{eq:int ||D||^3 << int der^6}.
The proofs for the last parts \ref{it:r^4DD^t << ER6} and \ref{it:tr (H^6)^3 << E^6R6}
of the present lemma are very similar and we omit them here.
\end{proof}

%
%
%


\end{document}